%% file: quipu30a_arvix.tex
\documentclass[10pt,journal,cspaper,compsoc]{IEEEtran}

\usepackage{amsmath}
\usepackage{amssymb}
\usepackage{mathtools}
\usepackage{xspace}
\usepackage{multirow}
\usepackage{rotating}
\usepackage[table]{xcolor}
\usepackage{graphicx}
\usepackage{url}
\usepackage{cite}
\usepackage[boxed]{algorithm2e}
\usepackage{listings}

\newtheorem{theorem}{Theorem}

\newtheorem{corollary}[theorem]{Corollary}
\newtheorem{invariant}{Invariant}

\newcounter{definition}
\newenvironment{definition}[1][]{\refstepcounter{definition}\par\medskip\noindent%
\textbf{Definition~\thedefinition. #1} \rmfamily}{\medskip}
\newcounter{example}
\newenvironment{example}[1][]{\refstepcounter{example}\par\medskip\noindent%
\textbf{Example~\theexample. #1} \rmfamily}{\medskip}
\newcounter{observation}
\newenvironment{observation}[1][]{\refstepcounter{observation}\par\medskip\noindent%
\textbf{Observation~\theobservation. #1} \rmfamily}{\medskip}

\newcommand{\lang}[1]{\texttt{#1}\xspace}

\newcommand{\chp}{\lang{CHP}}
\newcommand{\qcl}{\lang{QCL}}
\newcommand{\qpro}{\lang{QuIDDPro}}
\newcommand{\qplt}{\lang{QPLite}}
\newcommand{\quipu}{\lang{Quipu}}

\input{Qcircuit}
\usepackage[matrix,frame,arrow]{xypic}

\newcommand{\cC}{\mathcal{C}}
\newcommand{\cM}{\mathcal{M}}
\newcommand{\cF}{\mathcal{F}}

\makeatletter
\def\blfootnote{\xdef\@thefnmark{}\@footnotetext}
\makeatother

\definecolor{blue}{rgb}{0,0,0}

\begin{document}

\title{Simulation of Quantum Circuits \\ via Stabilizer Frames}
\author{\IEEEauthorblockN{H\'{e}ctor J. Garc\'{i}a \hspace{25pt} 
Igor L. Markov \\}
\IEEEauthorblockA{University of Michigan, EECS, Ann Arbor, MI 48109-2121 \\} 
\IEEEauthorblockA{\small \{hjgarcia, imarkov\}@eecs.umich.edu}
}


\maketitle

\begin{abstract}
Generic quantum-circuit simulation appears intractable for conventional computers and may be unnecessary because useful quantum circuits exhibit significant structure that can be exploited during simulation. For example, Gottesman and Knill identified an important subclass, called {\em stabilizer circuits}, which can be simulated efficiently using group-theory techniques and insights from quantum physics.
Realistic circuits enriched with quantum error-correcting codes and fault-tolerant procedures are dominated by stabilizer subcircuits and contain a relatively small number of non-Clifford components. Therefore, we develop new data structures and algorithms that facilitate
parallel simulation of such circuits. {\em Stabilizer frames} offer more compact storage than previous approaches but require more sophisticated bookkeeping. Our implementation, called \quipu, simulates certain quantum arithmetic circuits (e.g., reversible ripple-carry adders) in polynomial time and space for equal superpositions of $n$-qubits. On such instances, known linear-algebraic simulation techniques, such as the (state-of-the-art) BDD-based simulator \qpro, take exponential time. We simulate quantum Fourier transform and quantum fault-tolerant circuits using \quipu, and the results demonstrate that our stabilizer-based technique 
empirically outperforms \qpro in all cases. While previous 
high-performance, structure-aware simulations of quantum circuits were difficult to parallelize, 
we demonstrate that \quipu can be parallelized with a nontrivial computational speedup.
\end{abstract}



\section{Introduction}  
\label{sec:intro}

Quantum information processing manipulates quantum states rather than 
conventional $0$-$1$ bits. It has been demonstrated with a variety of physical 
technologies (NMR, ion traps, Josephson junctions in superconductors, optics) 
and used in recently developed commercial products. Examples of such products
include MagiQ's quantum key distribution system and ID-Quantique's
quantum random number generator. 
Shor's factoring algorithm~\cite{Shor} and Grover's search algorithm~\cite{Grover}
apply the principles of quantum information to carry out computation 
{\em asymptotically} more efficiently than conventional computers. 
These developments fueled research efforts to design, build and program 
scalable quantum computers. 
Due to the high volatility of quantum information, 
quantum error-correcting codes (QECC) and effective fault-tolerant (FT) architectures 
are necessary to build reliable quantum computers. 
For instance, the work in \cite{Oskin} describes practical FT architectures
for quantum computers, and \cite{Isai} explores architectures for 
constructing reliable communication channels via distribution of high-fidelity 
EPR pairs in a large quantum computer.
Most quantum algorithms are described in terms of {\em quantum circuits} and,
just like conventional digital circuits, require functional
simulation to determine the best FT design choices given limited resources.
In particular, {\em high-performance simulation} is a key component in quantum
design flows \cite{Svore} that facilitates analysis of
trade-offs between performance and accuracy.
Simulating quantum circuits on a conventional computer is a difficult problem. 
The matrices representing quantum gates, and the vectors that model quantum 
states grow exponentially with an increase in the number of {\em qubits}
-- the quantum analogue of the classical bit.
Several software packages have been 
developed for quantum-circuit simulation including Oemer's
Quantum Computation Language (\qcl) \cite{Oemer} and Viamontes'
Quantum Information Decision Diagrams (QuIDD) implemented 
in the \qpro package~\cite{Viamontes}. While \qcl simulates
circuits directly using state vectors, \qpro uses a
variant of binary decision diagrams to store state vectors more
compactly in some cases. Since the state-vector representation
requires excessive computational resources in general, 
simulation-based reliability studies (e.g. fault-injection analysis) 
of quantum FT architectures using general-purpose
simulators has been limited to small quantum circuits~\cite{Boncalo}. 
Therefore, {\em designing fast 
simulation techniques that target quantum FT circuits 
facilitates more robust reliability analysis of larger 
quantum circuits}.

\begin{figure}[!b]
    \small
    \[
        H = \frac{1}{\sqrt{2}}\begin{pmatrix}
            1 & 1 \\
            1 & -1 \end{pmatrix} \hspace{5pt}
        P = \begin{pmatrix}
            1 & 0 \\
            0 & i \end{pmatrix} \hspace{5pt}
        CNOT = \begin{pmatrix}
            1 & 0 & 0 & 0 \\
            0 & 1 & 0 & 0 \\
            0 & 0 & 0 & 1 \\
            0 & 0 & 1 & 0 \end{pmatrix}
    \]
    \vspace{-10pt}
    \caption{\label{fig:chp} {\color{blue} Clifford} gates: Hadamard (H), Phase (P) and controlled-NOT (CNOT).
 	}
\end{figure}

\ \\\noindent
{\bf Stabilizer circuits and states}.
Gottesman~\cite{Gottes98} and Knill identified an important subclass of quantum circuits, called 
{\em stabilizer circuits}, which can be simulated efficiently on classical
computers. Stabilizer circuits are exclusively composed of {\em {\color{blue} Clifford} gates} -- 
Hadamard, Phase and controlled-NOT gates (Figure~\ref{fig:chp}) followed by one-qubit measurements in the
computational basis. Such circuits are applied to a computational basis state
(usually $\ket{00...0}$) and produce output states known as {\em stabilizer states}.
Because of their extensive applications in QECC and
FT architectures, stabilizer circuits 
have been studied heavily \cite{AaronGottes, Gottes98}.
Stabilizer circuits can be simulated in {\em polynomial-time} by keeping track of the Pauli operators that stabilize\footnote{An operator $U$ is said to stabilize a state iff $U\ket{\psi}=\ket{\psi}$.} the
quantum state. Such {\em stabilizer operators} are maintained during 
simulation and uniquely represent stabilizer states up to an unobservable 
global phase.\footnote{According to quantum physics, 
the global phase $exp(i\theta)$ of a quantum state is unobservable
and does not need to be simulated.} Thus, this technique 
offers an {\em exponential improvement} over the computational resources  
needed to simulate stabilizer circuits using vector-based representations. 

Aaronson and Gottesman \cite{AaronGottes} proposed an improved 
technique that uses a bit-vector representation to simulate stabilizer circuits. 
Aaronson implemented this simulation approach in his \chp software package. Compared to
other vector-based simulators (\qpro, \qcl) the technique in \cite{AaronGottes} 
does not maintain the global phase of a state and simulates each {\color{blue} Clifford} gate 
in $\Theta(n)$ time using $\Theta(n^2)$ space. The overall runtime
of \chp is dominated by measurements, which require $O(n^2)$ time 
to simulate.

\ \\ \noindent
{\bf Stabilizer-based simulation of generic circuits}. We propose a generalization
of the stabilizer formalism that admits simulation of {\em {\color{blue} non-Clifford} gates}
such as Toffoli\footnote{The Toffoli gate is a $3$-bit gate that 
maps ($a$, $b$, $c$) to ($a$, $b$, $c \oplus (ab)$).} gates.
This line of research was first outlined in \cite{AaronGottes}, where the authors
describe a stabilizer-based representation that stores an arbitrary quantum state 
as a sum of density-matrix\footnote{Density matrices are self-adjoint positive-semidefinite matrices of trace $1.0$, that describe the statistical state of a quantum system \cite{NielChu}.} 
terms. In contrast, we {\em store arbitrary states as superpositions\footnote{
A superposition is a norm-$1$ linear combination of terms.} of pure
stabilizer states}. Such superpositions are stored more compactly than 
the approach from \cite{AaronGottes}, although we do not handle mixed stabilizer states.
The key obstacle to the more efficient pure-state approach has been the need to 
maintain the global phase of each stabilizer state in a superposition, where such 
phases become relative. We develop a new algorithm to overcome this obstacle.
We store stabilizer-state superpositions compactly using our
proposed {\em stabilizer frame} data structure.
To speed up relevant algorithms, we {\em store generator sets for each 
stabilizer frame in row-echelon form} to avoid expensive Gaussian elimination during simulation.
The main advantages of using stabilizer-state superpositions
to simulate quantum circuits are: 

\begin{itemize}
	\vspace{2pt}
	\item[(1)] {\em Stabilizer subcircuits are simulated with high efficiency.}
	\vspace{2pt}
	\item[(2)] {\em Superpositions can be restructured and compressed
	on the fly during simulation to reduce resource requirements.}
	\vspace{2pt}
	\item[(3)] {\em Operations performed on such superpositions
	can be computed in parallel and lend themselves
	to distributed or asynchronous processing.}
\end{itemize}

Our stabilizer-based technique simulates certain quantum arithmetic 
circuits in polynomial time and space 
for input states consisting of equal superpositions of 
computational-basis states. On such instances, 
well-known generic simulation techniques take exponential time.
We simulate various quantum Fourier transform
and quantum fault-tolerant circuits, and the results demonstrate that 
our data structure leads to orders-of-magnitude improvement in 
runtime and memory as compared to state-of-the-art simulators.

In the remaining part of this document, we assume at least a superficial familiarity 
with quantum computing. Section~\ref{sec:background} describes
key concepts related to quantum-circuit simulation and
the stabilizer formalism. In Section~\ref{sec:engr}, we introduce stabilizer frames 
and describe relevant algorithms. Section~\ref{sec:sframes} describes in detail
our simulation flow implemented in \quipu, and Section~\ref{sec:mthreads} 
discusses a parallel implementation of our technique. In Section~\ref{sec:results},
we describe empirical validation of \quipu and in single- and multi-threaded variants,
as well as comparisons with state-of-the-art simulators. Section~\ref{sec:conclude} 
closes with concluding remarks. 

\section{Background and Previous Work}  
\label{sec:background}

Quantum information processes, including quantum algorithms, are often
modeled using {\em quantum circuits} and, just like conventional digital 
circuits, are represented by diagrams~\cite{NielChu, Viamontes}.
Quantum circuits are sequences of {\em gate operations}
that act on some register of {\em qubits} -- the basic unit of information in a quantum
system. The quantum state $\ket{\psi}$ of a single qubit is described
by a two-dimensional complex-valued vector.
In contrast to classical bits, 
qubits can be in a {\em superposition} of the $0$ and $1$ states. Formally, $\ket{\psi}=\alpha_0\ket{0}+\alpha_1\ket{1}$, where 
$\ket{0} = (1,0)^\top$ and $\ket{1}=(0,1)^\top$ are the two-dimensional {\em computational 
basis states} and $\alpha_i$ are {\em probability amplitudes} that satisfy 
$|\alpha_0|^2 +|\alpha_1|^2 = 1$. An $n$-qubit register is the
tensor product of $n$ single qubits and thus is modeled by a complex vector $\ket{\psi^n}= \ket{\psi_1}\otimes\cdots\otimes\ket{\psi_n}=\sum_{i=0}^{2^n-1}\alpha_i\ket{b_i}$, 
where each $b_i$ is a binary string representing the value $i$
of each basis state. Furthermore, $\ket{\psi^n}$ satisfies $\sum_{i=0}^{2^n-1}|\alpha_i|^2 = 1$.
Each gate operation or {\em quantum gate} is a {\em unitary matrix} that operates on a small 
subset of the qubits in a register. For example, the quantum analogue of a NOT 
gate is the operator 
$X=\left(\begin{smallmatrix}0 & 1 \\ 1 & 0\end{smallmatrix}\right)$, 
\[\alpha_0\ket{00}+\alpha_1\ket{10} \xmapsto{X\otimes I} \alpha_0\ket{10}+\alpha_1\ket{00}\]
Similarly, the two-qubit CNOT operator flips the second qubit 
(target) iff the first qubit (control) is set to $1$, e.g., 
\[\alpha_0\ket{00}+\alpha_1\ket{10} \xmapsto{CNOT} \alpha_0\ket{00}+\alpha_1\ket{11}\]
Another operator of particular importance is the 
Hadamard (H), which is frequently
used to put a qubit in a superposition of computational-basis
states, e.g., 
\[\alpha_0\ket{00}+\alpha_1\ket{10} \xmapsto{I\otimes H} \frac{\alpha_0(\ket{00}+\ket{01}) +\alpha_1(\ket{10}+\ket{11})}{\sqrt{2}}\]
Note that the H gate generates {\em unbiased} superpositions in the sense that
the squares of the absolute value of the amplitudes are equal.

The dynamics involved in observing a quantum state are described by 
non-unitary {\em measurement operators}~\cite[Section 2.2.3]{NielChu}. There are 
different types of quantum measurements, but the type most pertinent to our discussion comprises 
{\em projective measurements in the computational basis}, i.e., measurements with respect
to the $\ket{0}$ or $\ket{1}$ basis states. The corresponding measurement operators
are $P_0=\left(\begin{smallmatrix} 1 & 0 \\ 0 & 0\end{smallmatrix}\right)$ and 
$P_1=\left(\begin{smallmatrix} 0 & 0 \\ 0 & 1\end{smallmatrix}\right)$,
respectively. The probability $p(x)$ of obtaining outcome $x\in\{0,1\}$
on the $j^{th}$ qubit of state $\ket{\psi}$ is given by the inner product 
$\bra{\psi}P_x^j\ket{\psi}$, where $\bra{\psi}$
is the conjugate transpose of $\ket{\psi}$. For example, the probability
of obtaining $\ket{1}$ upon measuring 
$\ket{\psi}=\alpha_0\ket{0}+\alpha_1\ket{1}$ is
\[p(1) = (\alpha_0^*, \alpha_1^*)P_1(\alpha_0, \alpha_1)^\top
= (0, \alpha_1^*)(\alpha_0, \alpha_1)^\top = |\alpha_1|^2\]

\ \\\noindent
{\bf Cofactors of quantum states}. The output states obtained after 
performing computational-basis measurements are called {\em cofactors}, 
and are states of the form $\ket{0}\ket{\psi_0}$ and $\ket{1}\ket{\psi_1}$.
These states are orthogonal to each other and add up
to the original state. The norms of cofactors and the original state
are subject to the Pythagorean theorem.
We denote the $\ket{0}$- and $\ket{1}$-cofactor by $\ket{\psi^{c=0}}$ and $\ket{\psi^{c=1}}$,
respectively, where $c$ is the index of the measured qubit. 
One can also consider {\em iterated cofactors}, such as
{\em double cofactors} $\ket{\psi^{qr=00}}$, $\ket{\psi^{qr=01}}$, $\ket{\psi^{qr=10}}$
and $\ket{\psi^{qr=11}}$. Cofactoring with respect to all qubits produces amplitudes 
of individual basis vectors. Readers familiar with cofactors of Boolean functions 
can use intuition from logic optimization and Boolean function theory.

\subsection{Quantum circuits and simulation} 
\label{sec:qcirc}

To simulate a quantum circuit $\cC$, we first initialize the quantum system to 
some desired state $\ket{\psi}$ (usually a basis state). $\ket{\psi}$ can be represented using
a fixed-size data structure (e.g., an array of $2^n$ complex numbers) or a variable-size 
data structure (e.g., algebraic decision diagram). We then track the 
evolution of $\ket{\psi}$ via its internal representation as the gates in $\cC$ are applied 
one at a time, eventually producing the output state $\cC\ket{\psi}$ \cite{AaronGottes, NielChu, Viamontes}. 
Most quantum-circuit simulators \cite{DeRaedt, Obenland, Oemer, Viamontes}
support some form of the linear-algebraic operations described 
earlier. The drawback of such simulators
is that their runtime and memory requirements grows exponentially in the number of qubits. 
This holds true not only in the worst case but also in practical
applications involving quantum arithmetic and quantum FT circuits. 


Gottesman developed 
a simulation method involving the {\em Heisenberg model}~\cite{Gottes98} 
often used by physicists to describe atomic-scale phenomena. {\em In this model,
one keeps track of the symmetries of an object rather than 
represent the object explicitly}. 
In the context of quantum-circuit simulation,
this model represents quantum states by their symmetries,
rather than complex-valued vectors and amplitudes.
The symmetries are operators 
for which these states are $1$-eigenvectors.
Algebraically, symmetries form {\em group} structures,
which can be specified compactly by group generators \cite{Hunger}. 

    
\subsection{The stabilizer formalism} 
\label{sec:stab}


A unitary operator $U$ {\em stabilizes} a state $\ket{\psi}$ iff
$\ket{\psi}$ is a $1$--eigenvector of $U$, i.e., $U\ket{\psi} 
= \ket{\psi}$. We are interested in operators $U$
derived from the Pauli matrices: $X=\left(\begin{smallmatrix} 0 & 1 \\ 1 & 0\end{smallmatrix}\right), 
Y=\left(\begin{smallmatrix} 0 & -i \\ i & 0\end{smallmatrix}\right), 
Z=\left(\begin{smallmatrix} 1 & 0 \\ 0 & -1\end{smallmatrix}\right)$, and the identity $I=\left(\begin{smallmatrix} 1 & 0 \\ 0 & 1\end{smallmatrix}\right)$.
The one-qubit states stabilized by the Pauli matrices are:
\begin{center}
	\begin{tabular}{lclc}
		$X$ : & $(\ket{0}+\ \ket{1})/\sqrt{2}$ & $-X$ : & $(\ket{0}-\ \ket{1})/\sqrt{2}$ \\
		$Y$ : & $(\ket{0}+i\ket{1})/\sqrt{2}$ & $-Y$ : & $(\ket{0}-i\ket{1})/\sqrt{2}$ \\
		$Z$ : & $\ket{0}$ & $-Z$ : & $\ket{1}$ \\
	\end{tabular}
\end{center}
\begin{table}[!b]
 	\parbox[t]{.8\linewidth}{
    \caption{\label{tab:pauli_mult} Multiplication table for Pauli matrices. Shaded cells
    indicate anticommuting products.}\vspace{-5pt}}
        \centering
        \begin{tabular}{|c||c|c|c|c|}
            \hline
                &   $I$ & $X$                         & $Y$   & $Z$ \\ \hline\hline
            $I$ &   $I$ & $X$                         & $Y$   &  $Z$ \\ \hline
            $X$ &   $X$ & $I$   & \cellcolor[gray]{0.85} $iZ$  & \cellcolor[gray]{0.85} $-iY$ \\ \hline
            $Y$ &   $Y$ & \cellcolor[gray]{0.85} $-iZ$ & $I$   & \cellcolor[gray]{0.85} $iX$ \\ \hline
            $Z$ &   $Z$ & \cellcolor[gray]{0.85} $iY$  & \cellcolor[gray]{0.85} $-iX$ & $I$ \\
            \hline
        \end{tabular}
\end{table}
Observe that $I$ stabilizes all states and $-I$ does not stabilize any state. 
Thus, the entangled state $(\ket{00} + \ket{11})/\sqrt{2}$ is stabilized by 
the Pauli operators $X\otimes X$, $-Y\otimes Y$, $Z\otimes Z$ and $I\otimes I$.  
As shown in Table \ref{tab:pauli_mult}, it turns out that the Pauli matrices 
along with $I$ and the multiplicative factors $\pm1$, $\pm i$, form a 
{\em closed group} under matrix multiplication \cite{NielChu}. Formally, the {\em Pauli group} 
$\mathcal{G}_n$ on $n$ qubits consists of the $n$-fold tensor product 
of Pauli matrices, $P = i^kP_1\otimes\cdot\cdot\cdot\otimes P_n$ 
such that $P_j\in\{I, X, Y, Z\}$ and $k\in\{0,1,2,3\}$. For brevity, the tensor-product symbol 
is often omitted so that $P$ is denoted by a string of $I$, $X$, $Y$ 
and $Z$ characters or {\em Pauli literals}, and a separate integer value
$k$ for the phase $i^k$. This string-integer pair representation allows us to compute 
the product of Pauli operators without explicitly 
computing the tensor products,\footnote{This holds true due to the 
identity: $(A\otimes B)(C \otimes D)=(AC\otimes BD)$.}
e.g., $(-IIXI)(iIYII) = -iIYXI$. Since $\mid \mathcal{G}_n\mid= 4^{n+1}$, 
$\mathcal{G}_n$ can have at most $\log_2 \mid \mathcal{G}_n \mid = 
\log_2 4^{n+1} = 2(n + 1)$ irredundant generators~\cite{NielChu}.
The key idea behind the stabilizer formalism is to represent an 
$n$-qubit quantum state $\ket{\psi}$ by its {\em stabilizer group} 
$S(\ket{\psi})$ -- the subgroup of $\mathcal{G}_n$ that stabilizes $\ket{\psi}$.

	\begin{theorem} \label{th:gen_commute}
        For an $n$-qubit pure state $\ket{\psi}$ and $k\leq n$, $S(\ket{\psi}) \cong {\mathbb Z}_2^k$. 
        If $k=n$, $\ket{\psi}$ is specified uniquely by $S(\ket{\psi})$ and is called a stabilizer state.
	\end{theorem}
	\begin{proof}
            (\emph{i}) To prove that $S(\ket{\psi})$ is commutative,
            let $P, Q \in S(\ket{\psi})$ such that $PQ\ket{\psi} = \ket{\psi}$. If $P$ and $Q$ anticommute,
            -$QP\ket{\psi} =$ -$Q(P\ket{\psi}) =$ -$Q\ket{\psi} =$ -$\ket{\psi} \neq \ket{\psi}$.
            Thus, $P$ and $Q$ cannot both be elements of $S(\ket{\psi})$.
			
			\noindent
            ({\em ii}) To prove that every element of $S(\ket{\psi})$ is of
            degree $2$, let $P \in S(\ket{\psi})$ such that $P\ket{\psi} = \ket{\psi}$.
            Observe that $P^2 = i^lI$ for $l\in\{0,1,2,3\}$.
            Since $P^2\ket{\psi} = P(P\ket{\psi}) = P\ket{\psi} = \ket{\psi}$, we obtain
            $i^l = 1$ and $P^2 = I$.
			
			\noindent
            ({\em iii}) {\color{blue} From group theory, a finite Abelian group 
            with identity element $e$ such that $a^2 = e$ for every 
            element $a$ in the group must be $\cong {\mathbb Z}_2^k$. }
			
			\noindent
            ({\em iv}) We now prove that $k \leq n$.
            First note that each independent generator $P \in S(\ket{\psi})$
            imposes the linear constraint $P\ket{\psi}=\ket{\psi}$
            on the $2^n$-dimensional vector space.
            The subspace of vectors that satisfy such a constraint
            has dimension $2^{n-1}$, or half the space. Let $gen(\ket{\psi})$
            be the set of generators for $S(\ket{\psi})$.
            We add independent generators to $gen(\ket{\psi})$ one by one and impose
  			their linear constraints, to limit $\ket{\psi}$ to the shared
			$1$-eigenvector. Thus the size of $gen(\ket{\psi})$ is at most $n$.
			In the case $|gen(\ket{\psi})| = n$, the $n$ independent 
			generators reduce the subspace of possible states to dimension
        		one. Thus, $\ket{\psi}$ is uniquely specified.
	\end{proof}

The proof of Theorem~\ref{th:gen_commute} shows that $S(\ket{\psi})$ 
is specified by only $\log_2 2^{n} = n$ {\em irredundant stabilizer generators}. 
Therefore, an arbitrary $n$-qubit stabilizer state can be represented by
a {\em stabilizer matrix} $\cM$ whose rows represent a set of 
generators $Q_1,\ldots,Q_n$ for $S(\ket{\psi})$. (Hence we use the terms 
{\em generator set} and {\em stabilizer matrix} interchangeably.) Since each 
$Q_i$ is a string of $n$ Pauli literals, the size of the matrix is $n\times n$. 
The fact that $Q_i\in S(\ket{\psi})$ implies that the leading phase of $Q_i$
can only be~$\pm 1$ and not~$\pm i$.\footnote{Suppose the phase
of $Q_i$ is $\pm i$, then $Q_i^2=$-$I \in S(\ket{\psi})$ which is not
possible since -$I$ does not stabilize any state.}
Therefore, we store the phases of each $Q_i$ separately using a binary 
vector of size $n$.

The storage cost for $\cM$ is $\Theta(n^2)$, which is an {\em exponential 
improvement} over the $O(2^n)$ cost often encountered in 
vector-based representations. 

\begin{example}\label{ex:stbmtx}
The state \begin{small}$\ket{\psi} = (\ket{00} + \ket{11})/\sqrt{2}$\end{small} 
is uniquely specified by any of the following matrices: 
\begin{small}
$\cM_1=\begin{smallmatrix}+\\+\end{smallmatrix}\left[\begin{smallmatrix} XX \\ ZZ \end{smallmatrix}\right]$,
\end{small}
\begin{small}
 $\cM_2=\begin{smallmatrix}+\\-\end{smallmatrix}\left[\begin{smallmatrix} XX \\ YY \end{smallmatrix}\right]$,
\end{small}
\begin{small}
 $\cM_3=\begin{smallmatrix}-\\+\end{smallmatrix}\left[\begin{smallmatrix} YY \\ ZZ \end{smallmatrix}\right]$.
\end{small}
One obtains $\cM_2$ from $\cM_1$ by left-multiplying the second row
by the first. Similarly, $\cM_3$ is obtained from $\cM_1$ 
or $\cM_2$ via row multiplication. 
Observe that multiplying any row by itself yields
$II$, which stabilizes $\ket{\psi}$. However, $II$ cannot
be used as a generator because it is redundant
and carries no information about the structure of $\ket{\psi}$. 
This holds true in general for $\cM$ of any size.
\end{example}
 	
Theorem~\ref{th:gen_commute} suggests that Pauli literals can be 
represented using two bits, e.g., $00 = I$, $01 = Z$, $10 = X$ and $11 = Y$. 
Therefore, a stabilizer matrix can be encoded using an $n\times2n$ binary matrix or {\em tableau}.
{\color{blue} This approach induces a linear map
${\mathbb Z}_2^{2n}\mapsto\mathcal{G}_n$ because vector addition
in ${\mathbb Z}_2^{2}$ is equivalent to multiplication of Pauli operators
up to a global phase.} The tableau implementation of the stabilizer formalism
is covered in \cite{AaronGottes, NielChu}. 
	\begin{observation} \label{obs:stabst_amps}
    		Consider a stabilizer state $\ket{\psi}$ represented by a set of generators of its
    		stabilizer group $S(\ket{\psi})$. Recall from the proof of Theorem~\ref{th:gen_commute}
    		that, since $S(\ket{\psi}) \cong {\mathbb Z}_2^n$,
   	 	each generator imposes a linear constraint on $\ket{\psi}$. Therefore, the set
        of generators can be viewed as a system of linear equations whose solution
        yields the $2^k$ (for some $k$ between $0$ and $n$) non-zero computational-basis 
        amplitudes that make up $\ket{\psi}$. Thus,
        one needs to perform Gaussian elimination to obtain such basis
        amplitudes from a generator set. 
    \end{observation}
\ \\\noindent
{\bf Canonical stabilizer matrices}. Observe from Example~\ref{ex:stbmtx} that,
although stabilizer states are uniquely determined by their stabilizer group, 
the set of generators may be selected in different ways. Any stabilizer matrix 
can be rearranged by applying sequences of elementary row operations in order 
to obtain the {\em row-reduced echelon form} structure depicted in Figure~\ref{fig:sminv}b. 
This particular stabilizer-matrix structure defines a {\em canonical representation} for 
stabilizer states~\cite{Djor, Gottes98}.
The elementary row operations that can be performed on a stabilizer matrix are
transposition, which swaps two rows of the matrix, and multiplication,
which left-multiplies one row with another. Such operations do not modify 
the stabilizer state and resemble the steps performed during a
Gaussian elimination\footnote{\scriptsize Since Gaussian elimination
essentially inverts the $n\times 2n$ matrix, this
could be sped up to $O(n^{2.376})$ time by using
fast matrix inversion algorithms. However, $O(n^3)$-
time Gaussian elimination seems more practical.} procedure.
Several row-echelon (standard) forms for stabilizer generators along 
with relevant algorithms to obtain them have been introduced in 
the literature~\cite{Audenaert, Gottes98, NielChu}. 
	
\ \\\noindent
{\bf Stabilizer-circuit simulation}. The computational-basis states 
are stabilizer states that can be represented by the 
stabilizer-matrix structure depicted in Figure~\ref{fig:sminv}a.
In this matrix form, the $\pm$ sign of each row along with 
its corresponding $Z_j$-literal designates whether 
the state of the $j^{th}$ qubit is $\ket{0}$ ($+$) or $\ket{1}$ ($-$). Suppose we 
want to simulate circuit $\cC$. Stabilizer-based
simulation first initializes $\cM$ to specify some basis 
state. Then, to simulate the action of each gate $U \in \cC$,
we conjugate each row $Q_i$ of $\cM$ by $U$.\footnote{
Since $Q_i\ket{\psi} = \ket{\psi}$, 
the resulting state $U\ket{\psi}$ is stabilized by $UQ_iU^\dag$ 
because $(UQ_iU^\dag) U\ket{\psi} = UQ_i\ket{\psi} = U\ket{\psi}$.} 
We require that the conjugation $UQ_iU^\dag$ maps to another
string of Pauli literals so that the resulting
matrix $\cM'$ is well-formed. It turns out that the H, P and CNOT gates have such mappings, 
i.e., these gates conjugate the Pauli group onto itself~\cite{Gottes98, NielChu}.
Table~\ref{tab:cliff_mult} lists mappings for the H, P and CNOT
gates.

\begin{figure}[!t]
	\centering
	\begin{tabular}{lclc}
		\hspace{-6pt}\rotatebox{90}{\hspace{10mm}\rotatebox{-90}{\bf (a)}} &
		\hspace{-12pt}\includegraphics[scale=.30]{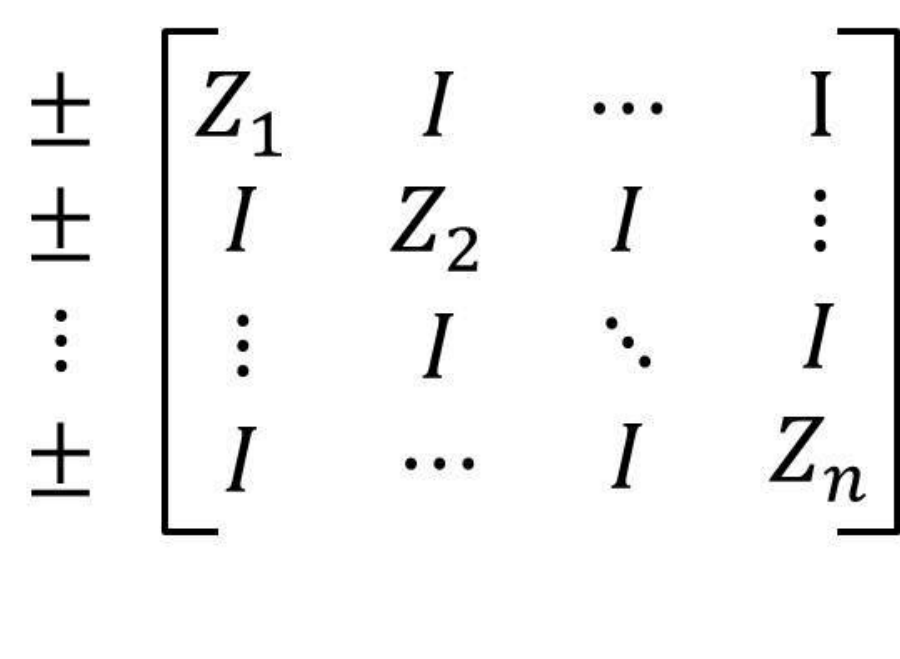}
		&
		\rotatebox{90}{\hspace{10mm}\rotatebox{-90}{\bf (b)}}&
		\hspace{-12pt} 
		\includegraphics[scale=.25]{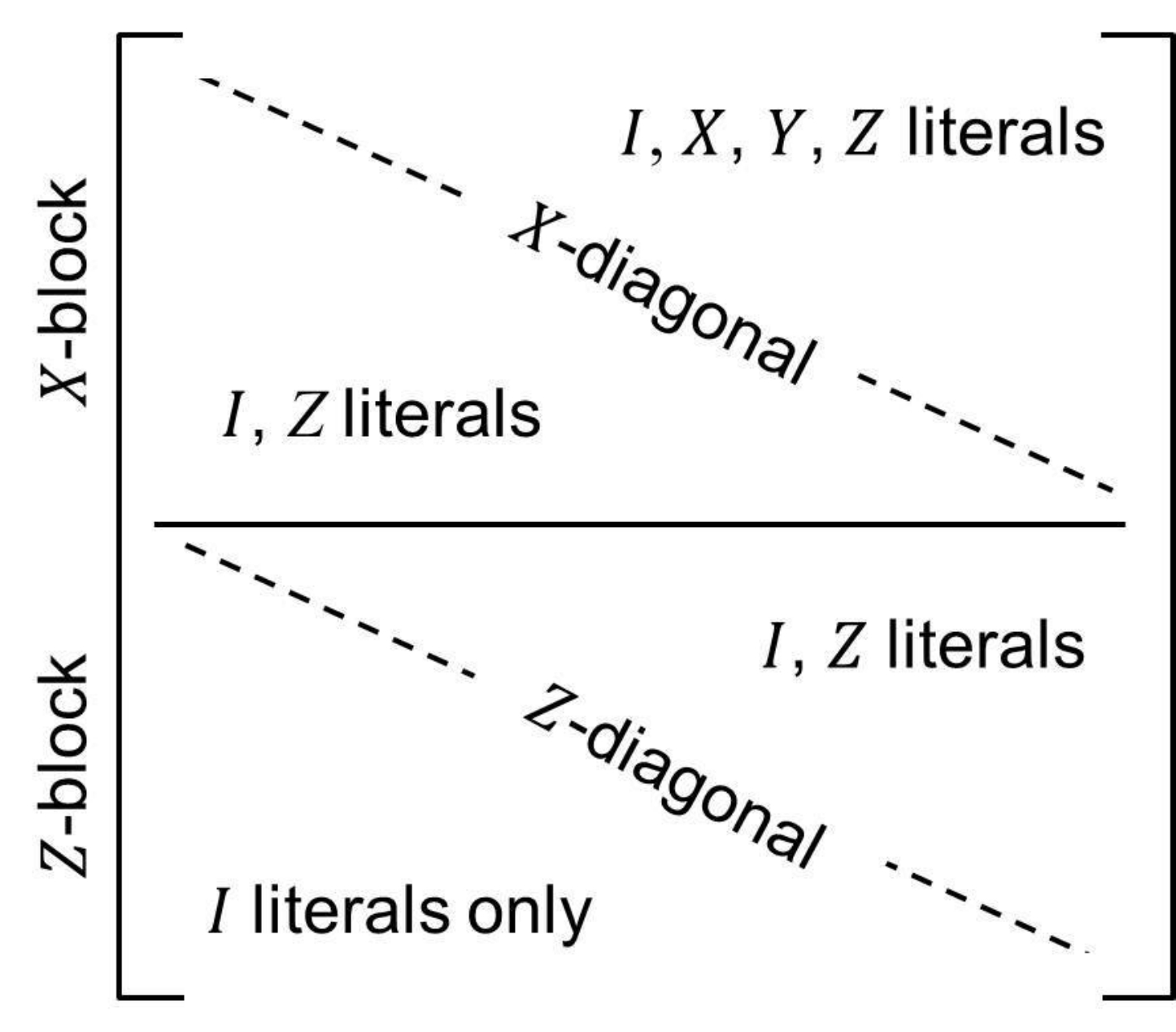} \\
	\end{tabular}
	\vspace{-5pt}
	\caption{\label{fig:sminv} {\bf (a)}~Stabilizer-matrix structure 
	for basis states. {\bf (b)}~Row-echelon form for stabilizer matrices. 
	The $X$-block contains a {\em minimal} set of generators with $X/Y$ literals. 
	Generators with $Z/I$ literals only appear in the $Z$-block.}
	\vspace{-10pt}
\end{figure}

	\begin{table}[!b]
        \centering 
        \parbox{.9\columnwidth}{\caption{\label{tab:cliff_mult} Conjugation of Pauli-group elements
    		by {\color{blue} Clifford} gates~\cite{NielChu}. For the $CNOT$ case, subscript $1$
    		indicates the control and $2$ the target.}\vspace{-8pt}}
        \begin{tabular}{cc}
        \begin{tabular}{|c||c|c|}
            \hline
                \sc Gate & \sc Input  & \sc Output \\ \hline\hline
                          & $X$ & $Z$ \\
                $H$       & $Y$ & -$Y$ \\
                          & $Z$ & $X$  \\ \hline
                          & $X$ & $Y$ \\
                $P$       & $Y$ & -$X$ \\
                          & $Z$ & $Z$   \\
            \hline
        \end{tabular}
        & \hspace{-5pt}
         \begin{tabular}{|c||c|c|}
            \hline
            \sc Gate & \sc Input  & \sc Output  \\ \hline\hline
            \multirow{6}{10mm}{$CNOT$} & $I_1X_2$ & $I_1X_2$   \\
                      & $X_1I_2$ & $X_1X_2$    \\
                      & $I_1Y_2$ & $Z_1Y_2$      \\
                      & $Y_1I_2$ & $Y_1X_2$     \\
                      & $I_1Z_2$ & $Z_1Z_2$   \\
                      & $Z_1I_2$ & $Z_1I_2$    \\
            \hline
        \end{tabular}
		\end{tabular}
 	\end{table}

\begin{example} Suppose we simulate a CNOT gate
on \begin{small}$\ket{\psi} = (\ket{00} + \ket{11})/\sqrt{2}$\end{small}. 
Using the stabilizer representation,
\begin{small}
$\cM_\psi=\left[\begin{smallmatrix} +XX \\ +ZZ \end{smallmatrix}\right]\xmapsto{CNOT}\cM'_\psi=\left[\begin{smallmatrix} +XI \\ +IZ \end{smallmatrix}\right]$.
\end{small}
%
The rows of $\cM'_\psi$ stabilize
\begin{small}$\ket{\psi}\xmapsto{CNOT}(\ket{00} + \ket{10})/\sqrt{2}$\end{small} as required.
\end{example}

Since H, P and CNOT gates are directly simulated
via the stabilizer formalism, these gates are {\color{blue} also known as
\emph{stabilizer gates}} and any circuit composed 
exclusively of such gates is called a unitary \emph{stabilizer
circuit}. Table~\ref{tab:cliff_mult} shows that at most two 
columns of $\cM$ are updated when a {\color{blue} Clifford (stabilizer)} gate is simulated.
Therefore, such gates are simulated in $\Theta(n)$ time. Furthermore, 
for any pair of Pauli operators $P$ and $Q$, $PQP^\dag=(-1)^cQ$, 
where $c=0$ if $P$ and $Q$ commute, and $c=1$ otherwise. Thus, 
Pauli gates can also be simulated in linear time as they
only permute the phase vector of the stabilizer matrix. 

	\begin{theorem} \label{th:stabst}
        An $n$-qubit stabilizer state $\ket{\psi}$ can be obtained
        by applying a stabilizer circuit to the $\ket{0}^{\otimes n}$
        computational-basis state. 
    \end{theorem}
    \begin{proof}
    		The work in \cite{AaronGottes} represents the generators 
    		using a tableau, and then shows how to construct a {\em canonical}
    		stabilizer circuit $\cC$ from the tableau. We refer the reader 
    		to~\cite[Theorem 8]{AaronGottes} for details of the proof.
    		Algorithms for obtaining more compact canonical circuits are
    		discussed in~\cite{Garcia}.
    \end{proof}

    \begin{corollary} \label{cor:stab_allzeros}
	   	An $n$-qubit stabilizer state $\ket{\psi}$ can be transformed
  		by {\color{blue} Clifford} gates into the $\ket{00\ldots 0}$ computational-basis state.
   	\end{corollary}
   	\begin{proof} 
        Since every stabilizer state can be produced by applying some unitary
        stabilizer circuit $\cC$ to the $\ket{0}^{\otimes n}$ state, it suffices to reverse
        $\cC$ to perform the inverse transformation. To reverse a stabilizer
        circuit, reverse the order of gates and replace every $P$ gate with $PPP$.
    \end{proof}
    
The stabilizer formalism also admits measurements 
in the computational basis \cite{Gottes98}. Conveniently, the formalism avoids
the direct computation of measurement operators and inner products (Section \ref{sec:background}).
However, the updates to $\cM$ for such gates are not as efficient
as for {\color{blue} Clifford} gates. Note that any qubit
$j$ in a stabilizer state is either in a $\ket{0}$ ($\ket{1}$) 
state or in an unbiased superposition of both.
The former case is called a {\em deterministic
outcome} and the latter a {\em random outcome}. We can tell these 
cases apart in $\Theta(n)$ time by searching for $X$ or $Y$ literals in
the $j^{th}$ column of $\cM$. If such literals are found, the qubit must be in 
a superposition and the outcome is random with equal probability
($p(0) = p(1) = .5$); otherwise the outcome is deterministic
($p(0) = 1$ or $p(1) = 1$). 

	{\em Case 1~--~randomized outcomes}: one flips an unbiased coin to 
	decide the outcome $x\in\{0,1\}$ and then updates $\cM$ to make 
	it consistent with the outcome. Let $R_j$ be a row in 
	$\cM$ with an $X/Y$ literal in its $j^{th}$ position, and let 
	$Z_j$ be the Pauli operator with a $Z$ literal in its $j^{th}$ 
	position and $I$ everywhere else. The phase of $Z_j$ is set to
	$+1$ if $x=0$ and $-1$ if $x=1$. Observe that $R_j$ and $Z_j$ 
	anticommute. If any other rows in $\cM$ anticommute with $Z_j$, 
	multiply them by $R_j$ to make them commute with $Z_j$. Then,
	replace $R_j$ with $Z_j$. Since this process requires up to
	$n$ row multiplications, the overall runtime is $O(n^2)$. 
	
	{\em Case 2~--~deterministic outcomes}: no updates to 
	$\cM$ are necessary but we need to figure out whether the qubit is 
	in the $\ket{0}$ or $\ket{1}$ state, i.e., whether the qubit is 
	stabilized by $Z$ or -$Z$. One approach is to perform Gaussian 
	elimination (GE) to put $\cM$ in row-echelon form. 
	This removes redundant literals from $\cM$ and makes it
	possible to identify the row containing a $Z$ in its $j^{th}$ position
	and $I$'s everywhere else. The $\pm$ phase of such a row 
	decides the outcome of the measurement. Since this is a GE-based
	approach, it takes $O(n^3)$ time in practice. 

The work in \cite{AaronGottes} improved the runtime of deterministic
measurements by doubling the size of $\cM$ to include $n$ {\em destabilizer generators}
in addition to the $n$ stabilizer generators. 
Such destabilizer generators help 
identify which specific row multiplications to compute in order 
to decide the measurement outcome. This approach avoids GE
and thus deterministic measurements are computed in $O(n^2)$ time. 
In Section~\ref{sec:engr}, we describe a different approach that computes 
such measurements in linear time without extra storage
but with an increase in runtime when simulating {\color{blue} Clifford} gates. 

In quantum mechanics, the states $e^{i\theta}\ket{\psi}$ and $\ket{\psi}$ are 
considered phase-equivalent because $e^{i\theta}$ does not affect the statistics 
of measurement. Since the stabilizer formalism simulates {\color{blue} Clifford} gates 
via their action-by-conjugation, such global phases are not maintained. 

\begin{example} Suppose we have state $\ket{1}$, which is stabilized by -$Z$.
Conjugating the stabilizer by the Phase gate yields $P($-$Z)P^\dag$=-$Z$.
However, in the state-vector representation, $\ket{1}\xmapsto{P}i\ket{1}$.
Thus the global phase $i$ is not maintained by the stabilizer. 
\end{example}

Since global phases are unobservable, they do not need to be maintained
when simulating a single stabilizer state. However, in 
Section~\ref{sec:sframes}, we show that such phases 
must be maintained when dealing with stabilizer-state
superpositions, where global phases become relative.


\section{Stabilizer Frames: Data Structure \\ and Algorithms}  
\label{sec:engr}

The {\color{blue} Clifford} gates by themselves do not form a universal set 
for quantum computation~\cite{AaronGottes, NielChu}. 
However, the Hadamard and Toffoli ($TOF$) gates do~\cite{Aharonov}.
To simulate $TOF$ and other {\em {\color{blue} non-Clifford} gates}, we extend
the formalism to include the representation of arbitrary 
quantum states as {\em superpositions of stabilizer states}.

\begin{example}\label{ex:superpos}
Recall from Section~\ref{sec:stab} that 
computational-basis states are stabilizer states. Thus, any one-qubit state 
$\ket{\psi} = \alpha_0\ket{0} + \alpha_1\ket{1}$ is 
a superposition of the stabilizer states $\ket{0}$ and $\ket{1}$. 
In general, any state decomposition in a computational basis is a 
stabilizer superposition.
\end{example}

{\color{blue}

Suppose $\ket{\psi}$ in Example~\ref{ex:superpos} is an {\em unbiased} state
such that $\alpha_0=i^k\alpha_1$ where $k = \{0,1,2,3\}$.
Then $\ket{\psi}$ can be represented using a single stabilizer state instead of two (up to a global phase).
} 
The key idea behind our technique is to identify and compress large unbiased 
superpositions on the fly during simulation to reduce resource requirements.
To this end, we leverage the following observation and derive a compact
data structure for representing stabilizer-state superpositions.

	\begin{observation}\label{obs:stbbasis}
		Given an $n$-qubit stabilizer state $\ket{\psi}$, there exists an
		orthonormal basis including $\ket{\psi}$ and consisting entirely of 
		stabilizer states. One such basis is obtained directly from the stabilizer
		of $\ket{\psi}$ by changing the signs of an arbitrary, non-empty subset of generators
		of $S(\ket{\psi})$, i.e., by modifying the phase vector of the
		stabilizer matrix for $\ket{\psi}$.\footnote{Let $S(\ket{\psi})$ 
		and $S(\ket{\varphi})$ be the stabilizer
		groups for $\ket{\psi}$ and $\ket{\varphi}$, respectively.
		If there exist $P \in S(\ket{\psi})$ and $Q \in S(\ket{\varphi})$
		such that $P =$ -$Q$, $\ket{\psi}$ and $\ket{\varphi}$
		are orthogonal since $\ket{\psi}$ is a $1$-eigenvector of $P$ and
		$\ket{\varphi}$ is a $(-1)$-eigenvector of $P$.}
		Thus, one can produce $2^n-1$ additional 
		orthogonal stabilizer states. Such states, together with 
		$\ket{\psi}$, form an orthonormal basis. This basis is 
		unambiguously specified by a single stabilizer state, 
		and any one basis state specifies the same basis.
	\end{observation}

\renewcommand\vec[1]{\ensuremath\boldsymbol{#1}}

	\begin{definition} \label{def:stab_frame}
		An $n$-qubit {\em stabilizer frame} $\cF$ is a set of $k\leq 2^n$ 
		stabilizer states $\{\ket{\psi_j}\}_{j=1}^{k}$ that forms an
		orthogonal subspace basis in the Hilbert space. 
		We represent $\cF$ by a pair consisting of 
		($i$) a stabilizer matrix $\cM$ and ($ii$) a set 
		of distinct {\em phase vectors} $\{\sigma_j\}_{j=1}^k$,
		where $\sigma_j\in\{\pm 1\}^n$. We use $\cM^{\sigma_j}$ to 
		denote the ordered assignment of the elements in $\sigma_j$ 
		as the ($\pm 1$)-phases of the rows in $\cM$. 
		Therefore, state $\ket{\psi_j}$ is represented by $\cM^{\sigma_j}$.
		The size of the frame, which we denote by $|\cF|$, is equal to $k$. 
	\end{definition}

Each phase vector $\sigma_j$ can be viewed as a binary ($0$-$1$) encoding 
of the integer index that denotes the respective basis vector. {\color{blue} Thus, when dealing with 
$64$ qubits or less}, a phase vector can be compactly represented by 
a $64$-bit integer (modern CPUs also support $128$-bit integers). 
To represent an arbitrary state $\ket{\Psi}$ using $\cF$, one 
additionally maintains a vector of complex amplitudes $\vec{a}=(a_1, \ldots, a_k)$, 
which corresponds to the decomposition of $\ket{\Psi}$ in the basis 
$\{\ket{\psi_j}\}_{j=1}^k$ defined by $\cF$, i.e., 
$\ket{\Psi} = \sum_{j=1}^ka_j\ket{\psi_j}$ and $\sum_{j=1}^k|a_j|^2=1$. 
Observe that each $a_j$ forms a pair with phase vector $\sigma_j$ in $\cF$
since $\ket{\psi_j}\equiv\cM^{\sigma_j}$. Any stabilizer state can be 
viewed as a one-element frame. 

	\begin{example}\label{ex:frame}
		Let $\ket{\Psi}=a_1(\ket{00}+\ket{01})+a_2(\ket{10}+\ket{11})$. Then
		$\ket{\Psi}$ can be represented by the stabilizer frame $\cF$ depicted 
		in Figure~\ref{fig:exfrm}.
	\end{example}

	\begin{figure}[!b]
		\centering
		\includegraphics[scale=.45]{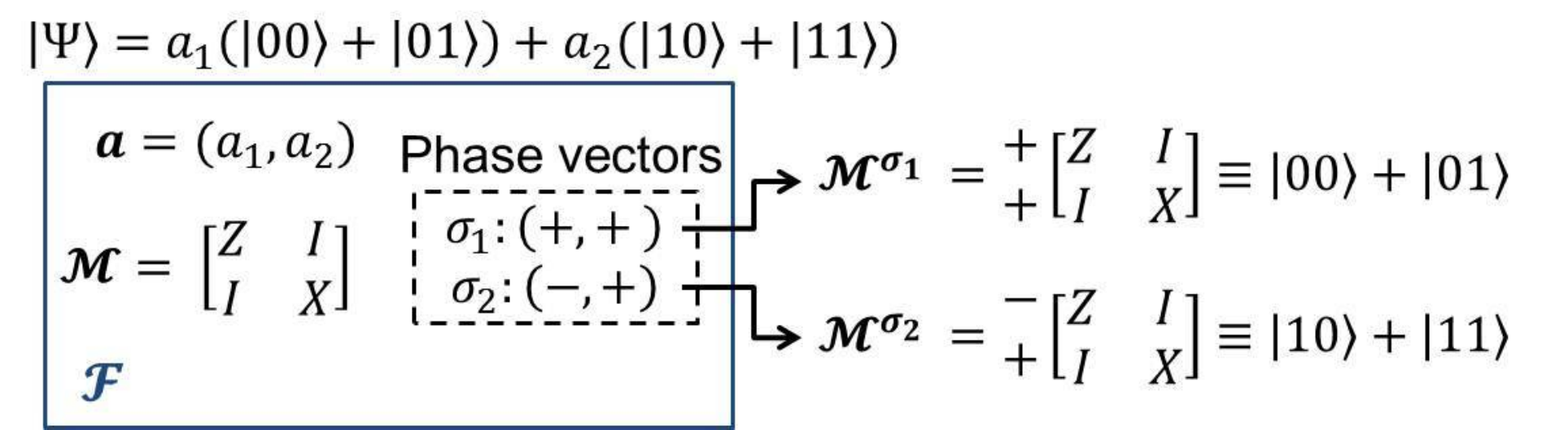}
		\caption{\label{fig:exfrm} A two-qubit stabilizer state
		$\ket{\Psi}$ whose frame representation uses two phase vectors.
		}
	\end{figure}

We now describe several {\em frame operations} that are useful
for manipulating stabilizer-state superpositions.

\newcommand{\frotate}{\mathrm{\tt ROTATE}}
\newcommand{\fcof}{\mathrm{\tt COFACTOR}}
\newcommand{\fmerge}{\mathrm{\tt MERGE}}
\newcommand{\fcoal}{\mathrm{\tt COALESCE}}

\ \\
\noindent
{$\mathbf\frotate(\cF, U)$}. Consider the stabilizer basis 
$\{\ket{\psi_j}\}_{j=1}^k$ defined by frame $\cF$. A stabilizer 
or Pauli gate $U$ acting on $\cF$ maps such a basis to 
$\{U\ket{\psi_j}=e^{i\theta_j}\ket{\varphi_j}\}_{j=1}^k$, 
where $e^{i\theta_j}$ is the global phase of stabilizer state $\ket{\varphi_j}$.
Since we obtain a new stabilizer basis that spans the same subspace, 
this operation effectively {\em rotates the stabilizer frame}.
Computationally, we perform a frame rotation as follows. First, update 
the stabilizer matrix associated with $\cF$ as per Section~\ref{sec:stab}. 
Then, iterate over the phase vectors in $\cF$ and update each one accordingly 
(Table~\ref{tab:cliff_mult}). 
Let $\vec{a}=(a_1, \ldots, a_k)\in\mathbb{C}^k$ be the decomposition of $\ket{\Psi}$ 
onto $\cF$. Frame rotation simulates the action of {\color{blue} Clifford} gate $U$
on $\ket{\Psi}$ since,
\begin{equation}\label{eq:gph}
U\ket{\Psi}=\sum_{j=1}^ka_jU\ket{\psi_j}=\sum_{j=1}^ka_je^{i\theta_j}\ket{\varphi_j}
\end{equation}
Observe that the global phase $e^{i\theta_j}$ of each 
$\ket{\varphi_j}$ becomes relative with respect to $U\ket{\Psi}$. 
Therefore, our approach requires that we compute such phases 
explicitly in order to maintain a consistent representation.   


\ \\\noindent
{\bf Global phases of states in $\cF$}. Recall from Theorem~\ref{th:stabst}
that any stabilizer state $\ket{\psi}=\cC\ket{0\ldots 0}$ for some 
stabilizer circuit $\cC$. To compute 
the global phase of $\ket{\psi}$, one keeps track of 
the global factors generated when each gate in $\cC$ 
is simulated in sequence. In the context of frames, we maintain 
the global phase of each state in $\cF$ using the amplitude vector $\mathbf{a}$.
Let $\cM$ be the matrix associated with $\cF$ and let $\mathbf{\sigma_j}$ be 
the ($\pm$)-phase vector in $\cF$ that forms a pair with $a_j\in \mathbf{a}$.
When simulating {\color{blue} Clifford} gate $U$, each $a_j$ is updated as follows: 

	\begin{algorithm}\small
		{\bf 1)}~Set the leading phases of the rows in $\cM$ to $\mathbf{\sigma_j}$. \\
		{\bf 2)}~Obtain a basis state $\ket{b}$ from $\cM$ and store 
		its non-zero amplitude $\beta$. If $U$ is the Hadamard gate, it
		may be necessary to sample a sum of two non-zero
		basis amplitudes (one real, one imaginary). \\
		{\bf 3)}~Consider the matrix representation of $U$ and the vector
		representation of $\beta\ket{b}$, and compute $U(\beta\ket{b})=\beta'\ket{b'}$ 
		via matrix-vector multiplication. \\
		{\bf 4)}~Obtain $\ket{b'}$ from $U\cM U^\dag$ and 
		store its amplitude~$\gamma\neq 0$. \\
		{\bf 5)}~Compute the global factor generated
		as $a_j=(a_j\cdot\beta')/\gamma$.
	\end{algorithm}
	
By Observation~\ref{obs:stabst_amps}, $\cM$ needs to be in row-echelon 
form (Figure \ref{fig:sminv}b) in order to sample the computational-basis
amplitudes $\ket{b}$ and $\ket{b'}$.
Thus, simulating gates with global-phase maintenance 
would take $O(n^3|\cF|)$ time for $n$-qubit stabilizer
frames. To improve this, we introduce a simulation invariant.
	\begin{invariant} \label{inv:sminv}
		The stabilizer matrix $\cM$ associated with $\cF$ 
		remains in row-echelon form during simulation.
	\end{invariant}
	
Since {\color{blue} Clifford} gates affect at most two columns of $\cM$,
Invariant~\ref{inv:sminv} can be repaired with $O(n)$ row multiplications. 
Since each row multiplication takes $\Theta(n)$, 
the runtime required to update $\cM$ during 
global-phase maintenance simulation is $O(n^2)$. 
Therefore, for an $n$-qubit stabilizer frame, the overall runtime 
for simulating a single {\color{blue} Clifford} 
gate is $O(n^2 + n|\cF|)$ since one can memoize the updates to 
$\cM$ required to compute each $a_j$. 
Another advantage of maintaining this 
invariant is that the outcome of 
deterministic measurements (Section~\ref{sec:stab}) can be
decided in time linear in $n$ since it eliminates the need
to perform Gaussian elimination.

{\color{blue}

\ \\\noindent
{$\mathbf\fcof(\cF, c)$}. This operation facilitates measurement of
stabilizer-state superpositions and simulation of non-Clifford
gates using frames. (Recall from Section~\ref{sec:background} 
that post-measurement states are also called cofactors.) Here, 
$c\in\{1,2,\ldots, n\}$ is the cofactor index. Let 
$\{\ket{\psi_j}\}_{j=1}^k$ be the stabilizer basis defined by $\cF$. 
Frame cofactoring maps such a basis to $\{\ket{\psi_j^{c=0}},\ket{\psi_j^{c=1}}\}_{j=1}^k$.
Therefore, after a frame is cofactored, its size either
remains the same (qubit $c$ was in a deterministic state
and thus one of its cofactors is empty) or doubles (qubit $c$ was in 
a superposition state and thus both cofactors are non-empty). 
We now describe the steps required to cofactor~$\cF$. 

	\begin{algorithm}\small
		{\bf 1)}~Check $\cM$ associated with $\cF$ to determine 
		whether qubit $c$ is a random (deterministic) state. (Section~\ref{sec:stab})\\
		{\bf 2a)}~Suppose qubit $c$ is in a deterministic state. Since $\cM$ is maintained 
		in row-echelon form (Invariant~\ref{inv:sminv}) no frame updates 
		are necessary. \\
		{\bf 2b)}~In the randomized-outcome case, apply the measurement algorithm 
		from Section~\ref{sec:stab} to $\cM$ while forcing the outcome to 
		$x\in\{0,1\}$. This is done twice~--~once for each $\ket{x}$-cofactor, 
		and the row operations performed on $\cM$ are memoized each time. \\
		{\bf 3)}~Let $\mathbf{\sigma_j}$ be the ($\pm$)-phase vector in $\cF$ 
		that forms a pair with $a_j$. Iterate over each $\langle\mathbf{\sigma_j}, a_j\rangle$
		pair and update its elements according to the memoized operations. \\
		{\bf4)}~For each $\langle\mathbf{\sigma_j}, a_j\rangle$, insert a new phase
		vector-amplitude pair corresponding to the cofactor state added to the stabilizer basis.
	\end{algorithm}

\noindent
Similar to frame rotation, the runtime of cofactoring is linear in the number of phase vectors
and quadratic in the number of qubits. However, after this operation, the number 
of phase vectors (states) in $\cF$ will have grown by a (worst case) factor of two.
Furthermore, any state $\ket{\Psi}$ represented by $\cF$ is 
invariant under frame cofactoring.

\begin{example}
	Figure~\ref{fig:toffcof} shows how $\ket{\Psi}=\ket{000}+\ket{010}+\ket{100}+\ket{110}$ 
	is cofactored with respect to its first two qubits. A total of four cofactor states are 
	obtained.
\end{example}

} 

\section{Simulating Quantum Circuits \\ with Stabilizer Frames} 
\label{sec:sframes}

Let $\cF$ be the stabilizer frame used to represent the
$n$-qubit state $\ket{\Psi}$. Following our discussion 
in Section~\ref{sec:engr}, any stabilizer 
or Pauli gate can be simulated directly via frame rotation. 

Suppose we want to simulate the action of $TOF_{c_1c_2t}$, where $c_1$ and $c_2$
are the control qubits, and $t$ is the target.
First, we decompose $\ket{\Psi}$ into all four of its double cofactors
(Section~\ref{sec:background}) over the control qubits to obtain
the following equal superposition of orthogonal states:
\begin{equation*}
	\ket{\Psi} = \frac{\ket{\Psi^{c_1c_2=00}} + \ket{\Psi^{c_1c_2=01}}
	+ \ket{\Psi^{c_1c_2=10}} + \ket{\Psi^{c_1c_2=11}}}{2}
\end{equation*}
{\color{blue} 
Here, we assume the most general case where all
the $c_1c_2$ cofactors are non-empty. The number 
of states in the superposition obtained could also be two 
(one control qubit has an empty cofactor) or one 
(both control qubits have empty cofactors). After cofactoring,
we compute the action of the Toffoli as,
}
\begin{align}\label{eq:toff}
	TOF_{c_1c_2t}\ket{\Psi} =(&\ket{\Psi^{c_1c_2=00}} + \ket{\Psi^{c_1c_2=01}} \notag \\
	& + \ket{\Psi^{c_1c_2=10}} + X_t\ket{\Psi^{c_1c_2=11}})/2 
\end{align}
where $X_t$ is the Pauli gate (NOT) acting on target $t$. 
We simulate Equation~\ref{eq:toff} with the following 
frame operations. (An example of the process is depicted in 
Figure~\ref{fig:toffcof}.)

\begin{algorithm}\small
{\bf 1)} $\fcof(\cF, c_1)$. \\
{\bf 2)} $\fcof(\cF, c_2)$. \\
{\bf 3)} Let $Z_j$ be the Pauli operator with a $Z$ literal in 
its $j^{th}$ position and $I$ everywhere else. Due 
to Steps 1 and 2, the matrix $\cM$ associated with $\cF$ must
have two rows of the form $Z_{c_1}$ and $Z_{c_2}$. Let $u$ and $v$ 
be the indices of such rows, respectively. For each phase vector 
$\sigma_{j\in\{1,\ldots,|\cF|\}}$, \\ if the $u$ and $v$ elements 
of $\sigma_j$ are both -$1$ (i.e., if the phase vector 
corresponds to the \begin{footnotesize}$\ket{\Psi^{c_1c_2=11}}$\end{footnotesize} cofactor), flip the 
value of element $t$ in $\sigma_j$ (apply $X_t$ to this cofactor).
\end{algorithm}

\begin{figure}[!b]
	\centering\hspace{-.1cm}
	\includegraphics[scale=.38]{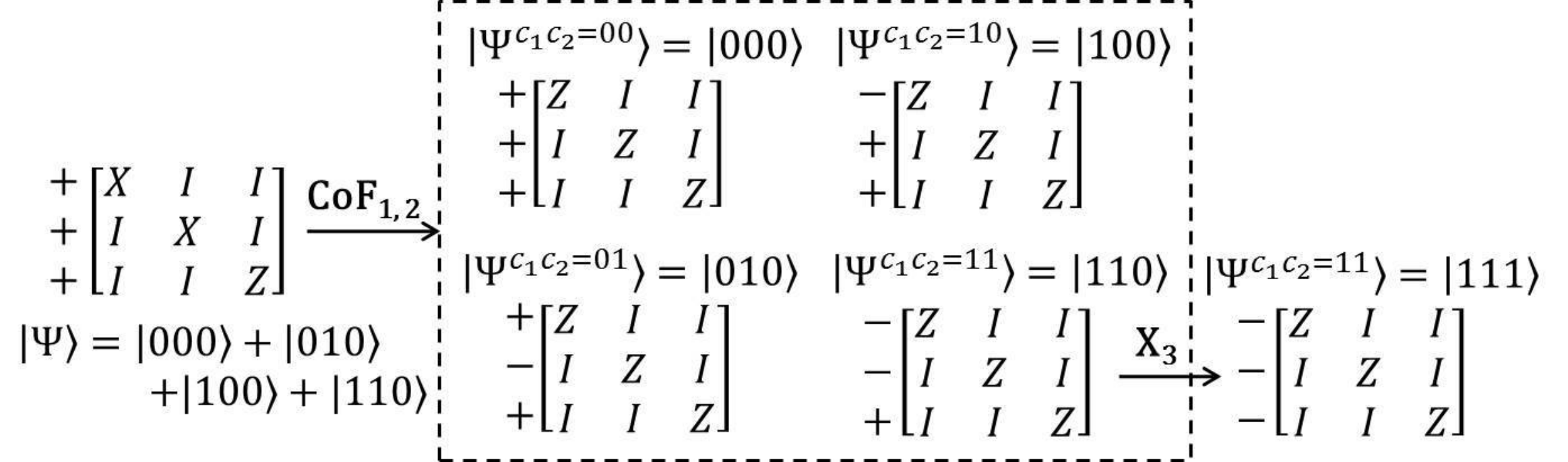}
	\caption{\label{fig:toffcof} Simulation of $TOF_{c_1c_2t}\ket{\Psi}$
	using a stabilizer-state superposition (Equation~\ref{eq:toff}). 
	Here, $c_1=1$, $c_1=2$ and $t=3$. Amplitudes are omitted for clarity
	and the $\pm$ phases are prepended to matrix rows. The $X$ gate is 
	applied to the third qubit of the $\ket{\Psi^{c_1c_2=11}}$ cofactor.}
\end{figure}

Controlled-phase gates $R(\alpha)_{ct}$ can also be simulated using
stabilizer frames. This gate applies a phase-shift factor of 
$e^{i\alpha}$ if both the control qubit $c$ and target qubit $t$
are set. Thus, we compute the action of $R(\alpha)_{ct}$ as,
\begin{align} \label{eq:crot}
	R(\alpha)_{ct}\ket{\Psi} =(&\ket{\Psi^{ct=00}} + \ket{\Psi^{ct=01}} \notag \\
		& + \ket{\Psi^{ct=10}} + e^{i\alpha}\ket{\Psi^{ct=11}})/2
\end{align}
Equation~\ref{eq:crot} can be simulated via frame-based simulation
using a similar approach as discussed for $TOF$ gates. Let 
$(a_1, \ldots, a_{|\cF|})$ be the decomposition of 
$\ket{\Psi}$ onto $\cF$. First, cofactor $\cF$ over the $c$ and $t$ qubits. 
Then, for any phase vector $\sigma_{j\in\{1,\ldots,|\cF|\}}$ that 
corresponds to the $\ket{\Psi^{ct=11}}$ cofactor, set $a_j = a_je^{i\alpha}$.
Observe that, in contrast to $TOF$ gates, controlled-$R(\alpha)$ gates
produce biased superpositions. The Hadamard and controlled-$R(\alpha)$ 
gates are used to implement the quantum Fourier transform circuit, 
which plays a key role in Shor's factoring algorithm.

\ \\
\noindent
{\bf Measuring $\cF$}. Since the states in $\cF$ are orthogonal,
the outcome probability when measuring $\cF$
is calculated as the sum of the normalized outcome probabilities 
of each state. The normalization is with respect to the 
amplitudes stored in $\mathbf{a}$, and thus the overall 
measurement outcome may have a non-uniform distribution.
Formally, let $\ket{\Psi} = \sum_i a_i\ket{\psi_i}$ be the
superposition of states represented by $\cF$, the 
probability of observing outcome $x \in \{0,1\}$ upon
measuring qubit $m$ is,
	\begin{equation*}
		p(x)_\Psi = \sum_{i=1}^k|a_i|^2\bra{\psi_i}P^m_x\ket{\psi_i}=\sum_{i=1}^k|a_i|^2 p(x)_{\psi_i}
	\end{equation*}
\noindent
where $P^m_x$ denotes the measurement operator in the
computational basis $x$ as discussed in Section~\ref{sec:background}.
The outcome probability for each stabilizer state $p(x)_{\psi_i}$ is computed
as outlined in Section~\ref{sec:stab}. 
Once we compute $p(x)_\Psi$, we flip a (biased) coin to decide 
the outcome and cofactor the frame such that only the states that are
consistent with the measurement remain in the frame. 

Prior work on simulation of {\color{blue} non-Clifford} gates using the stabilizer 
formalism can be found in~\cite{AaronGottes} where the authors represent 
a quantum state as a sum of $O(4^{2dk})$ density-matrix terms
while simulating $k$ non-Clifford operations acting on $d$ distinct 
qubits. In contrast, the number of 
states in our technique is $O(2^k)$ although we do not handle 
density matrices and perform more sophisticated bookkeeping. 

{\color{blue}

\ \\\noindent
{\bf Technology-dependent gate decompositions}. Stabilizer-frame simulation can be 
specialized to quantum technologies that are represented by libraries of 
primitive gates (and, optionally, macro gates) with restrictions on qubit 
interactions as well as gate parallelism and scheduling. The work
in~\cite{Lin} describes several primitive gate libraries for different quantum
technologies including quantum dot, ion trap and superconducting systems. 
Such libraries can be incorporated into frame-based simulation by 
decomposing the primitive gates into linear combinations of Pauli or 
Clifford operators, subject to qubit-interaction constraints. 
For example, the $T=\left(\begin{smallmatrix} 1 & 0 \\ 0 & e^{i\pi/4}\end{smallmatrix}\right)$
gate and its inverse, which are primitive gates in most quantum machine descriptions,
can be simulated as $T_t\ket{\Psi} =(\ket{\Psi^{t=0}} + e^{\pm i\pi/4}\ket{\Psi^{t=1}})/\sqrt{2}$.
}

\ \\\noindent
{\bf Multiframe simulation}. Although a single frame 
is sufficient to represent a stabilizer-state superposition $\ket{\Psi}$, 
one can sometimes tame the exponential growth of states in $\ket{\Psi}$ by constructing 
a {\em multiframe representation}. Such a representation cuts down the
total number of states required to represent $\ket{\Psi}$, 
thus improving the scalability of our technique. Our experiments 
in Section~\ref{sec:results} show that, when 
simulating ripple-carry adders, the number
of states in $\ket{\Psi}$ grows linearly when multiframes 
are used but exponentially when a single frame is used.
To this end, we introduce an additional frame operation.

\ \\
\noindent
{$\mathbf\fcoal(\cF)$}. One derives a multiframe representation directly 
from a single frame $\cF$ by examining the set of phase vectors and 
identifying {\em candidate pairs} that can be {\em coalesced} into a 
single phase vector associated with a different stabilizer matrix. 
Since we maintain the stabilizer matrix $\cM$ of a frame in row-echelon
form (Invariant~\ref{inv:sminv}), examining the phases corresponding to 
$Z_k$-rows ($Z$-literal in $k^{th}$ column and $I$'s in all other columns) 
allows us to identify the columns in $\cM$ that need to be modified in 
order to coalesce candidate pairs. More generally, suppose $\langle\sigma_r, \sigma_j\rangle$
are a pair of phase vectors from the same $n$-qubit frame. Then
$\langle\sigma_r, \sigma_j\rangle$ is considered a candidate iff 
it has the following properties: ({\em i}) $\sigma_r$ and $\sigma_j$ 
are equal up to $m\leq n$ entries corresponding to $Z_k$-rows 
(where $k$ is the qubit the row stabilizes), 
and ({\em ii}) $a_r=i^{d}a_j$ for some $d\in\{0,1,2,3\}$
(where $a_r$ and $a_j$ are the frame amplitudes paired 
with $\sigma_r$ and $\sigma_j$). 
Let $\vec{e}=\{e_1,\ldots,e_{m}\}$ be the indices of 
a set of differing phase-vector elements, and let
$\vec{v}=\{v_1,\ldots,v_{m}\}$ be the qubits stabilized 
by the $Z_k$-rows identified by $\vec{e}$. The steps in our 
coalescing procedure are:
\begin{algorithm}\small
	{\bf 1)}~Sort phase vectors such that {\em candidate pairs}
	with differing elements $\vec{e}$ are next to each other.\\
	{\bf 2)}~Coalesce candidates into a new set 
	of phase vectors $\vec{\sigma'}$. \\
	{\bf 3)}~Create a new frame $\cF'$ consisting of 
	$\vec{\sigma'}$ and matrix \begin{footnotesize}$\cC\cM\cC^\dag$\end{footnotesize},
	where\begin{footnotesize}
	$\cC$=CNOT$_{v_1, v_2}$CNOT$_{v_1, v_3}\cdots$CNOT$_{v_1, v_m}$P$^d_{v_1}$H$_{v_1}$\end{footnotesize}. \\
	{\bf 4)}~Repeat Steps 2--3 until no candidate pairs remain.
\end{algorithm}

\newcommand{\flst}{\mathbf{F}}

{\color{blue}

During simulation, we execute the {\em coalescing} operation after the set of
phase vectors in a frame is expanded via cofactoring. Therefore,
the choice of $\vec{e}$ (and thus $\vec{v}$) is driven by the $Z_k$-rows produced after a frame
is cofactored. (Recall that cofactoring modifies $\cM$ such that each 
cofactored qubit $k$ is stabilized by a $\pm Z_k$ operator.)
The output of our coalescing operation is a list of 
$n$-qubit frames $\flst=\{\cF'_1, \cF'_2, \ldots, \cF'_s\}$ (i.e., a multiframe) 
that together represent the same superposition as the original input frame $\cF$. 
The size of the multiframe generated is half the number of phase vectors
in the input frame. The runtime of this procedure is dominated by Step~1. Each phase-vector 
comparison takes $\Theta(n)$ time. Therefore, the runtime of Step~1 and our overall coalescing
procedure is $O(nk\log k)$ for a single frame with $k$ phase vectors. 

\begin{example}\label{ex:coalesce}
Suppose we coalesce the frame $\cF$ depicted in Figure~\ref{fig:multifrm}.
Candidate pairs are $\langle\sigma_1, \sigma_2\rangle$ and $\langle\sigma_3, \sigma_4\rangle$,
with $\vec{e}=\{2\}$ and $\vec{e}=\{2,3\}$, respectively.
To obtain $\cF_1$, conjugate the second column of $\cM$ by 
an H gate (Step~3), which will coalesce $\langle\sigma_1, \sigma_2\rangle$
into a single phase vector $\sigma_1$. 
Similarly, to obtain $\cF_2$, conjugate the second column 
by H, then conjugate the second and third columns by CNOT, 
which will coalesce $\langle\sigma_3, \sigma_4\rangle$. Observe that no P
gates are applied since $d=0$ for all pairs in $\vec{a}$.
\end{example}

} 

\begin{figure}[!t]
	\centering
	\includegraphics[scale=.36]{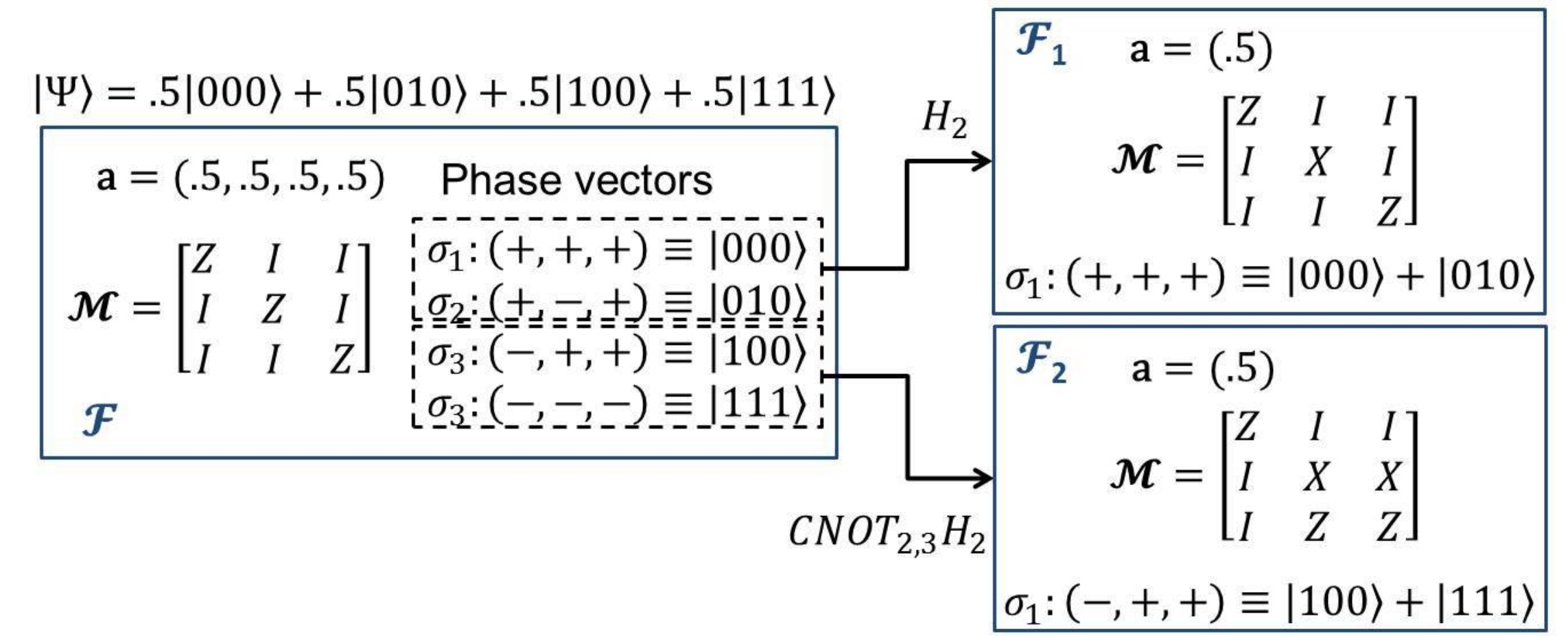}
	\vspace{-5pt}
	\caption{\label{fig:multifrm} Example of how a multiframe 
	representation is derived from a single-frame representation.
	}
\end{figure}


Candidate pairs can be identified even in the absence of $Z_k$-rows in 
an $n$-qubit $\cM$. By Corollary~\ref{cor:stab_allzeros}, one can always find 
a stabilizer circuit $\cC$ that maps $\cM$ to the matrix structure depicted in Figure~\ref{fig:sminv}a,
whose rows are all of $Z_k$ form. Several $O(n^3)$-time algorithms exist for obtaining 
$\cC$ \cite{AaronGottes, Audenaert, Garcia}. We leverage such 
algorithms to extend our coalescing operation as follows:
\begin{algorithm}\small
	{\bf 1)}~Find $\cC$ that maps $\cM$ to computational-basis form.\\
	{\bf 2)}~$\frotate(\cF, \cC)$.\\
	{\bf 3)}~$\{\cF'_1, \cF'_2, \ldots, \cF'_s\}\leftarrow\fcoal(\cF)$. \\
	{\bf 4)}~$\frotate(\cF'_i, \cC^\dag)$ for $i\in\{1,\ldots,s\}$.
\end{algorithm}

To simulate stabilizer, Toffoli and controlled-$R(\alpha)$
gates using multiframe $\flst$, we apply single-frame operations 
to each frame in the list independently. For
Toffoli and controlled-$R(\alpha)$ gates,  
additional steps are required: 
\begin{algorithm}\small
	{\bf 1)}~Apply the coalescing procedure to each frame 
	and insert the new ``coalesced'' frames in the list. \\
	{\bf 2)}~Merge frames with equivalent stabilizer matrices. \\
	{\bf 3)}~Repeat Steps 1--2 until no new frames are generated. 
\end{algorithm}

\noindent
{\bf Orthogonality of multiframes}. We introduce the following invariant 
to facilitate simulation of quantum measurements on multiframes.
	\begin{invariant} \label{inv:ortho}
		The stabilizer frames that represent a superposition
		of stabilizer states remain mutually orthogonal during simulation,
		i.e., every pair of (basis) vectors from any two frames 
		are orthogonal.
	\end{invariant}

\newcommand{\mlst}{\mathbf{M}}
\newcommand{\plst}{\mathbf{P}}

Given multiframe $\flst=\{\cF_1,\dots,\cF_k\}$, one needs to consider two
separate tasks in order to maintain Invariant~\ref{inv:ortho}.
The first task is to verify the pairwise orthogonality of
the states in $\flst$. The orthogonality of two
$n$-qubit stabilizer states can be checked using the inner-product algorithm 
describe in~\cite{Garcia}, which takes $O(n^3)$ time. To improve this,
we derive a heuristic based on 
Observation~\ref{obs:stbbasis}, which takes advantage
of similarities across the (canonical) matrices in $\flst$ 
to avoid expensive inner-product computations in many cases. 
We note that, when simulating quantum circuits that exhibit significant 
structure, $\flst$ contains similar stabilizer matrices 
with equivalent rows (Pauli operators).  
Let $\mlst=\{\cM_1,\dots,\cM_k\}$
be the set of $n$-qubit stabilizer matrices in $\flst$. Our heuristic 
keeps track of a set of Pauli operators $\plst=\{P_1,P_2,\ldots,P_{k\leq n}\}$, 
that form an {\em intersection} across the matrices in $\mlst$.

\begin{example}
	Consider the multiframe from Figure~\ref{fig:multifrm}. The
	intersection $\plst$ consists of $ZII$
	(first row of both $\cM$). 
\end{example}

By Observation~\ref{obs:stbbasis}, if two phase vectors (states)
have different entries corresponding to the Pauli operators 
in $\plst$, then the states are orthogonal and no inner-product
computation is required. For certain
practical instances, including the benchmarks described in 
Section~\ref{sec:results}, we obtain a non-empty $\plst$
and our heuristic proves effective. 
When $\plst$ is empty or the phase-vector pair is
equivalent, we use the algorithm
from \cite{Garcia} to verify orthogonality. Therefore, 
in the worst case, checking pairwise orthogonality
of the states in $\flst$ takes $O(n^3k^2)$ time for
a multiframe that represents a $k$-state superposition. 

The second task to consider when maintaining 
Invariant~\ref{inv:ortho} is the orthogonalization
of the states in $\flst=\{\cF_1,\dots,\cF_k\}$ when 
our check fails. To accomplish this, we iteratively apply the
$\fcof$ operation to each frame in $\flst$ in order
to decompose $\flst$ into a single frame. At each iteration, 
we select a {\em pivot qubit} $p$ based on the composition 
of Pauli literals in the corresponding column. We apply the 
$\fcof(\cF, p)$ operation only if there exists a pair of 
matrices in $\flst$ that contain a different set of Pauli 
literals in the pivot column.
\begin{algorithm}\small
	{\bf 1)}~Find pivot qubit $p$. \\
	{\bf 2)}~$\fcof(\cF_i, p)$ for $i\in\{1,\ldots,k\}$. \\
	{\bf 3)}~Merge frames with equivalent stabilizer matrices. \\
	{\bf 4)}~Repeat Steps 1--3 until a single frame remains.
\end{algorithm}

\begin{figure}[!b]
	\centering
	\includegraphics[scale=.24]{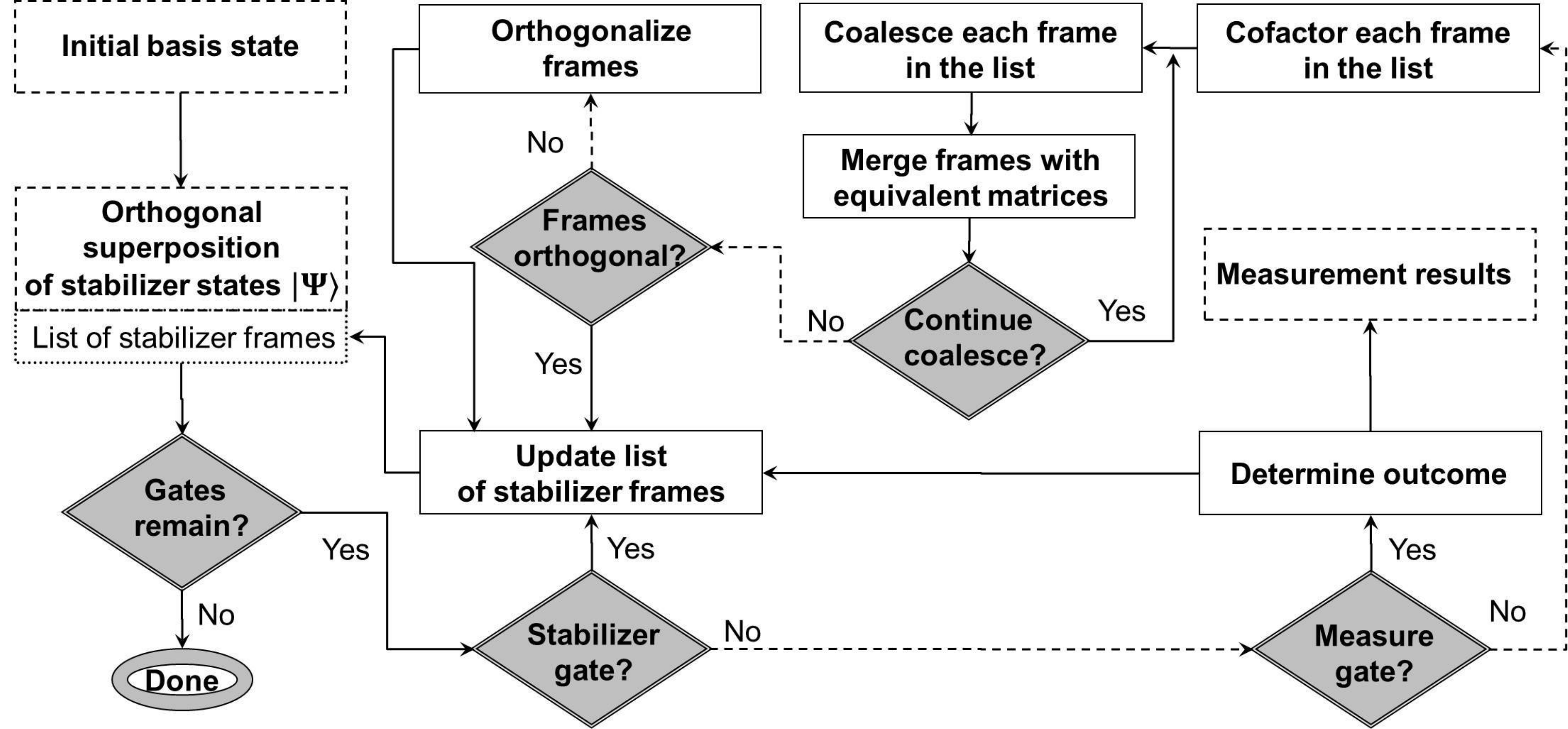} 
	\caption{\label{fig:flow}Overall simulation flow for \quipu.}
\end{figure}

{\color{blue}
Observe that the order of pivot selection does not 
impact the performance of orthogonalization. 
Each iteration of the algorithm can potentially double the number 
of states in the superposition. Since the algorithm terminates when a single frame remains, 
the resulting states are represent by distinct phase vectors and are therefore 
pairwise orthogonal. 
In the worst case, all $n$ qubits are selected as pivots and the resulting
frame constitutes a computational-basis decomposition of size $2^n$.

The overall simulation flow of our frame-based techniques is shown 
in Figure~\ref{fig:flow} and implemented in 
our software package \quipu. 
\begin{example} 
	Figure~\ref{fig:exflow} depicts the main steps of the \quipu 
	simulation flow for a small non-Clifford circuit.
\end{example}
}

\begin{figure*}[!t]
	\centering
	\includegraphics[scale=.32]{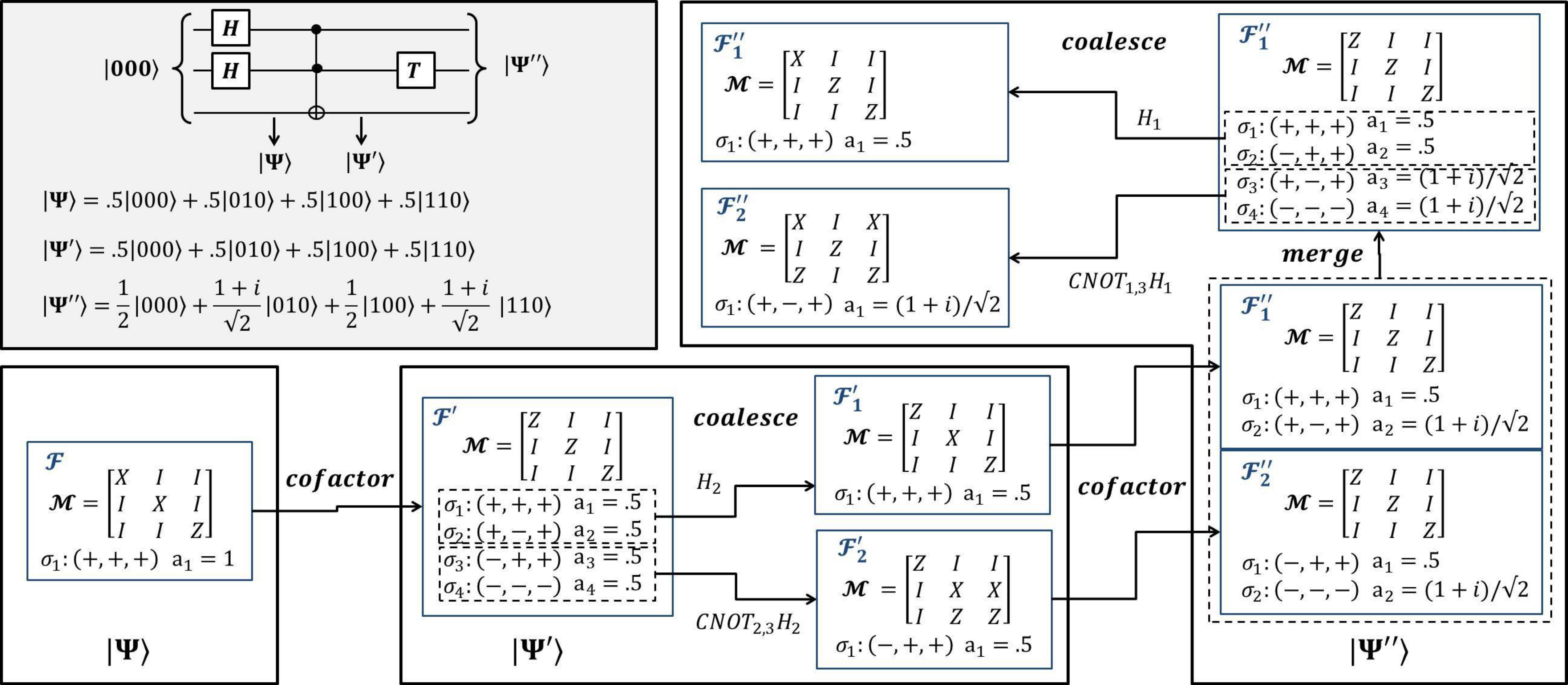} 
	\vspace{5pt}
	\caption{\label{fig:exflow} Example simulation flow for a small non-Clifford circuit (top left) using \quipu.
	The multiframes obtained are pairwise orthogonal and thus no orthogonalization is required.}
	\vspace{-10pt}
\end{figure*}

\section{Parallel Simulation} 
\label{sec:mthreads}

Unlike other techniques based on compact representations of quantum states 
(e.g., using BDD data structures~\cite{Viamontes}), most frame-based operations are inherently 
parallel and lend themselves to multi-threaded implementation. 
The only step in Figure~\ref{fig:flow} 
that presents a bottleneck for 
a parallel implementation is the orthogonalization procedure,
which requires communication across frames. 
All other processes at both the single- and multi-frame
levels can be executed on large subsets 
of phase vectors independently. 

\begin{figure}[!b]
\centering
\begin{lstlisting}[frame=single]
template<class Func, class Iter, class... Params>
auto async_launch( Func f, const Iter begin, const Iter end, 
                  Params... p ) 
  -> vector< decltype( async(f, begin, end, p...) ) >
{
  vector< decltype( async(f, begin, end, p...) ) > futures;
  int size = distance(begin, end);
  int n = size/MTHREAD;
  futures.reserve(MTHREAD);
  for(int i = 0; i < MTHREAD; i++)	{
    Iter first = begin + i*n;
    Iter last = (i < MTHREAD - 1) ? begin + (i+1)*n : end;
    futures.push_back( async( f, first, last, p...) );
  }
  return futures;
}
\end{lstlisting}
\caption{\label{fig:async} Our C++11 template function for parallel
execution of the frame operations described in Section~\ref{sec:sframes}. 
}
\end{figure}

We implemented a multithreaded version of \quipu using the
C++11 thread support library. Each worker thread is launched via
the {\ttfamily std::async()} function. Figure~\ref{fig:async}
shows our wrapper function for executing calls to 
{\ttfamily std::async()}. The {\ttfamily async\_launch()} function
takes as input: ($i$) a frame operation ({\ttfamily Func f}),
($ii$) a range of phase-vector elements defined by 
{\ttfamily Iter begin} and {\ttfamily Iter end}, and ($iii$) any
additional parameters ({\ttfamily Params...~p}) required for
the frame operation. Furthermore, the function returns a
vector of {\ttfamily std::future}~--~the C++11 mechanism for 
accessing the result of an asynchronous operation scheduled 
by the C++ runtime support system. The workload 
(number of phase vectors) of each thread is distributed 
evenly across the number of cores in the system ({\ttfamily MTHREAD}). 
The results from each thread are joined only when orthogonalization
procedures are performed since they require communication between
multiple threads. This is accomplished by calling the 
{\ttfamily std::future::get()} function on each future. 
All {\color{blue} Clifford} gates and measurements are simulated in parallel.

\section{Empirical Validation} 
\label{sec:results}

We tested a single-threaded and multi-threaded versions of \quipu
on a conventional Linux server using several benchmark sets
consisting of stabilizer circuits, quantum ripple-carry adders,
quantum Fourier transform circuits and quantum fault-tolerant (FT) circuits.

{\color{blue} We used a straightforward implementation of the state-vector 
model using an array of complex amplitudes to perform functional verification of: 
($i$)~all benchmarks with $< 30$ qubits and ($ii$)~\quipu output for such benchmarks. 
We simulated such circuits and checked
for equivalence among the resultant states and operators~\cite{Viamontes}.
In the case of stabilizer circuits, we used the equivalence-checking method
described in~\cite{AaronGottes, Garcia}.
}

	\begin{figure}[!b]
	\begin{tabular}{lclc}
		\hspace{-.4cm}
		\rotatebox{90}{\small\hspace{4mm}Avg. Runtime (secs)} & 
		\hspace{-.85cm} \includegraphics[scale=.38]{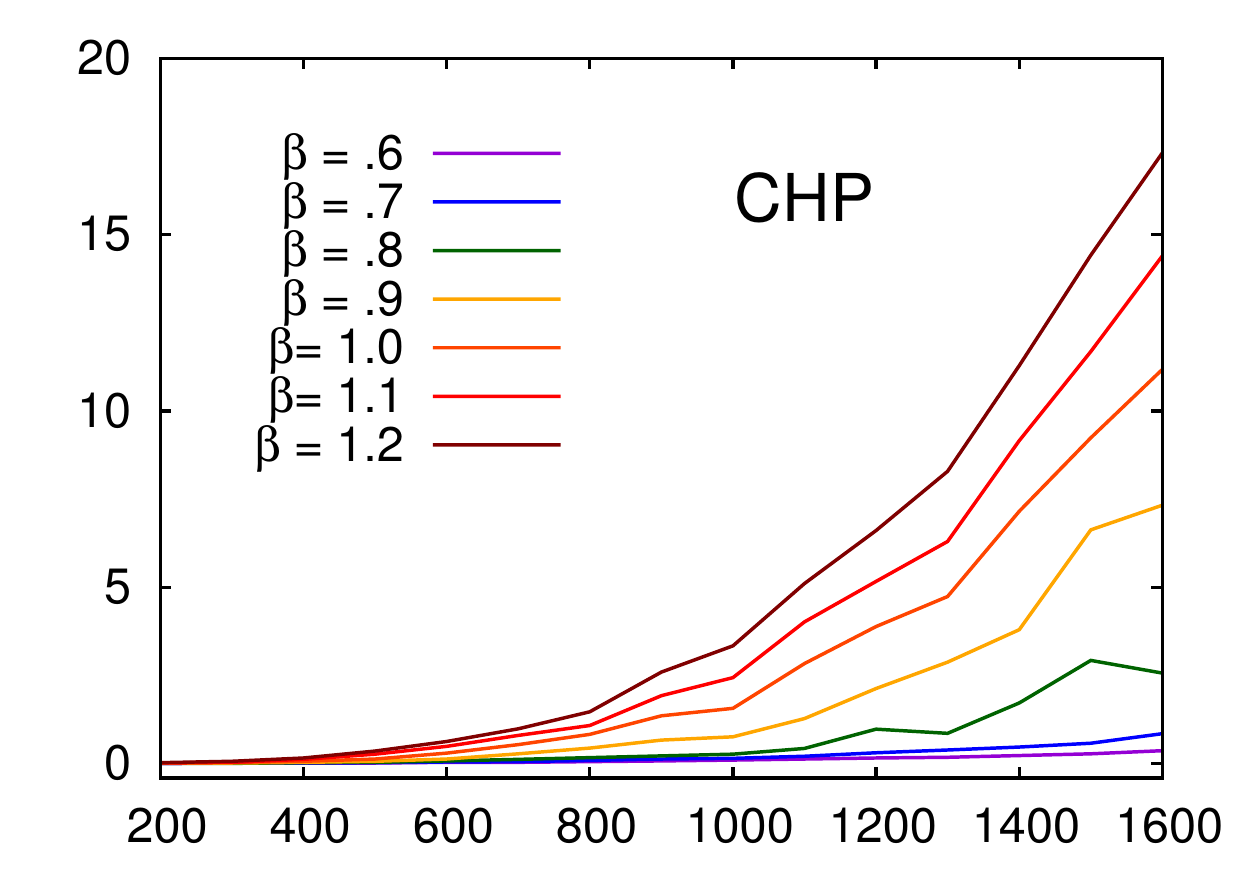} 
		\hspace{-.8cm} \includegraphics[scale=.38]{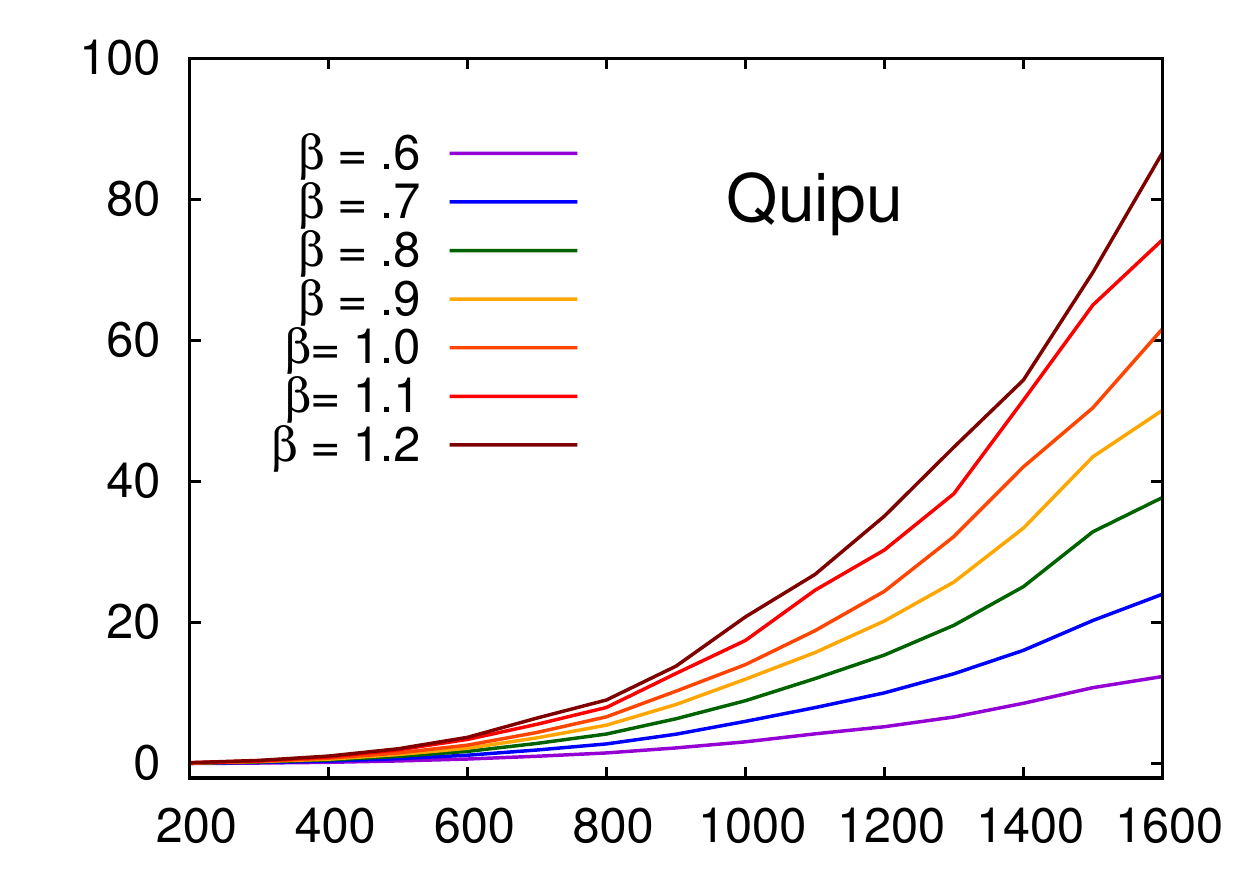} \\
		& \multicolumn{2}{c}{\hspace{2.5cm}\small Number of qubits}
	\end{tabular}
	\caption{\label{fig:stabcomp} Average time needed by \quipu and \chp to simulate an 
	$n$-qubit stabilizer circuit with $\beta n\log n$ gates and $n$ measurements. 
    \quipu is asymptotically as fast as \chp but is not limited
    to stabilizer circuits.
        }
	\end{figure}
	
\ \\\noindent
{\bf Stabilizer circuits}. We compared the runtime performance of single-threaded \quipu 
against that of \chp using a benchmark set similar to the one used in \cite{AaronGottes}. 
We generated random stabilizer circuits on $n$ qubits, for $n\in\{100, 200, \ldots, 1500\}$.
The use of randomly generated benchmarks is justified for our experiments because (\emph{i})~our 
algorithms are not explicitly sensitive to circuit topology and (\emph{ii})~random stabilizer circuits have
been considered representative \cite{Knill}.
For each $n$, we generated the circuits as follows: 
fix a parameter $\beta > 0$; then choose $\beta \lceil n \log_2 n\rceil$ 
random unitary gates (CNOT, P or H) each
with probability $1/3$. Then measure each qubit 
$a \in \{0,\ldots, n-1 \}$ in sequence.
We measured the number of seconds needed to simulate
the entire circuit. The entire procedure was repeated for $\beta$
ranging from $0.6$ to $1.2$ in increments of $0.1$. 
Figure \ref{fig:stabcomp} shows the average time needed by \quipu 
and \chp to simulate this benchmark set. 
The purpose of this 
comparison is to evaluate the overhead of supporting generic 
circuit simulation in \quipu. 
Since \chp is specialized to 
stabilizer circuits, we do not expect \quipu to be faster.
When $\beta = 0.6$, the simulation time appears
to grow roughly linearly in $n$ for both simulators. However, 
when the number of unitary gates is doubled
($\beta = 1.2$), the runtime of both simulators grows roughly quadratically. 
Therefore, the performance of both simulators depends strongly on the circuit being simulated. 
Although \quipu is $5\times$ slower than \chp, we note that \quipu maintains
global phases whereas \chp does not. Nonetheless, Figure~\ref{fig:stabcomp}
shows that \quipu is asymptotically as fast as \chp when simulating 
stabilizer circuits that contain a linear number of measurements.
The multithreaded speedup in \quipu for non-Clifford circuits is 
not readily available for stabilizer circuits.

\ \\\noindent
{\bf Ripple-carry adders}. Our second benchmark set consists of
$n$-bit ripple-carry (Cuccaro) adder \cite{Cuccaro} circuits, which 
often appear as components in many arithmetic circuits \cite{MarkovSaeedi}.
The Cuccaro circuit for $n=3$ is shown in Figure \ref{fig:rcadder}. Such circuits
act on two $n$-qubit input registers,
one ancilla qubit and one carry qubit for a total of $2(n+1)$ qubits. 
We applied H gates to all $2n$ input qubits in order to simulate 
addition on a superposition of $2^{2n}$ computational-basis states. 
Figure~\ref{fig:rcaddres} shows the average runtime needed to simulate
this benchmark set using \quipu. For comparison, we 
ran the same benchmarks on an optimized version of \qpro, 
called \qplt\footnote{\qplt is up to 
$4\times$ faster since it removes
overhead related to \qpro's interpreted front-end for
quantum programming.}, 
specific to circuit simulation \cite{Viamontes}.  
When $n<15$, \qplt is faster than \quipu
because the QuIDD representing the state vector remains compact
during simulation. However, for $n>15$, 
the compactness of the QuIDD is considerably reduced, 
and the majority of \qplt's runtime is spent in non-local 
pointer-chasing and memory (de)allocation. Thus, \qplt fails to scale on such 
benchmarks and one observes an exponential increase in runtime. 
Memory usage for both \quipu and \qplt was nearly unchanged
for these benchmarks. \quipu consumed $4.7$MB on average while
\qplt consumed almost twice as much ($8.5$MB). 
	
	\begin{figure}[!t]
		\centering
		\includegraphics[scale=.65]{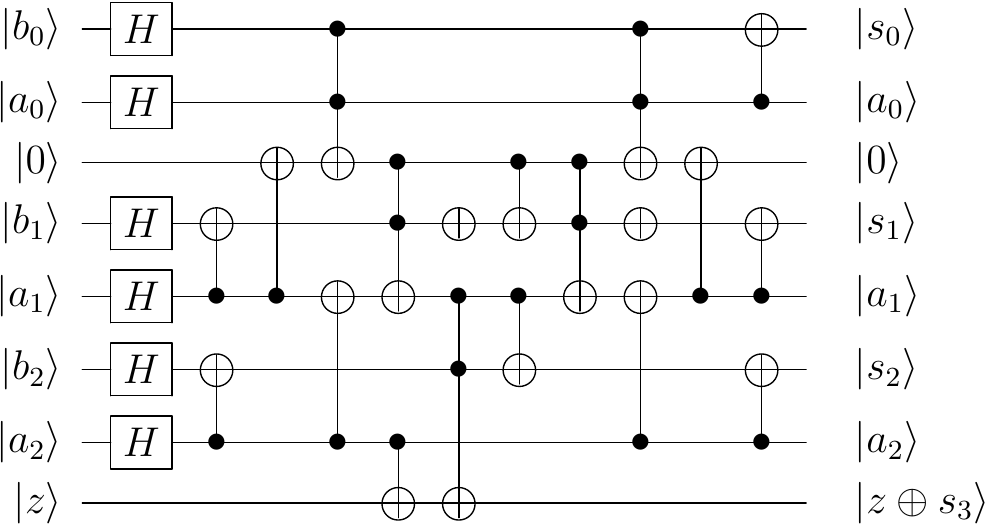}
		\vspace{4pt}
		\caption{\label{fig:rcadder} Ripple-carry (Cuccaro) adder 
		for $3$-bit numbers $a = a_0a_1a_2$
		and $b = b_0b_1b_2$ \cite[Figure 6]{Cuccaro}. 
		The third qubit from the top is an
		ancilla and the $z$ qubit is the carry.
		The $b$-register is overwritten with the result $s_0s_1s_2$.}
	\end{figure}
	
	\begin{figure}[!b]
	\centering
	\begin{tabular}{lc}
		\hspace{-.25cm}
		\rotatebox{90}{\small\hspace{1.5cm} Runtime (secs)} 
		& \hspace{-.83cm}
		\vspace{-.15cm}
		\includegraphics[scale=.50]{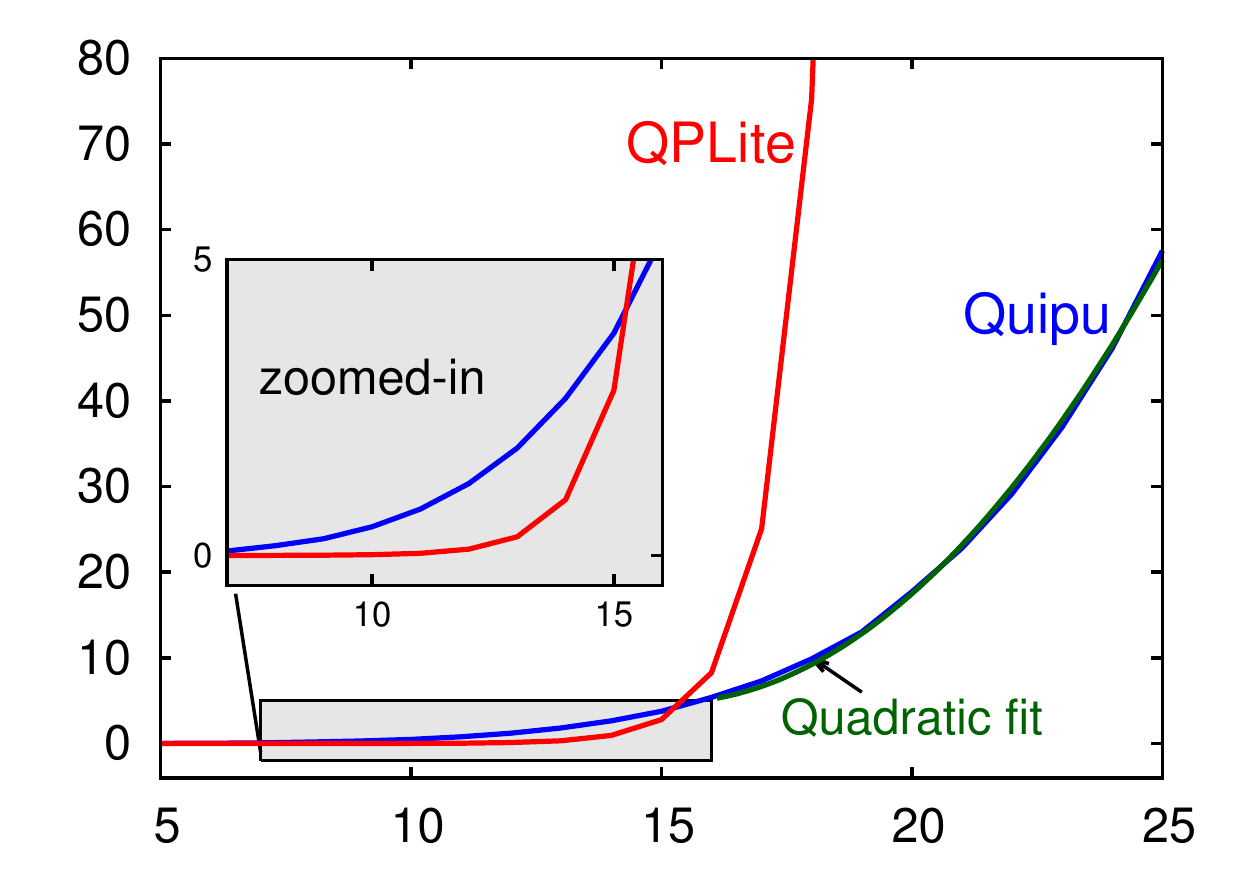} \\
		& \small$n$-bit Cuccaro adder ($2n+2$ qubits)
	\end{tabular}
	\parbox{\linewidth}{
	\caption{\label{fig:rcaddres} Average runtime and memory needed by \quipu 
	and \qpro to simulate $n$-bit Cuccaro adders on an equal superposition of all
	computational basis states. 
        }}
	\end{figure}
	
We ran the same benchmarks using single and multiframes. 
The number of states in the superposition grows exponentially in $n$
for a single frame, but linearly in $n$ when multiple frames are allowed.
This is because $TOF$ gates produce large equal superpositions
that are effectively compressed by our coalescing technique. 
Since our frame-based algorithms require poly($k$) time 
for $k$ states in a superposition, 
\quipu simulates Cuccaro circuits in polynomial time and space 
for input states consisting of large superpositions of basis states. 
On such instances, known linear-algebraic simulation techniques 
(e.g., \qpro) take exponential time while \quipu's runtime
grows quadratically (best quadratic fit $f(x) = 0.5248x^2 - 15.815x + 123.86$ 
with $R^2=.9986$).

The work in \cite{MarkovSaeedi} describes additional quantum arithmetic 
circuits that are based on Cuccaro adders (e.g., subtractors, conditional 
adders, comparators). We used \quipu to simulate such circuits and observed
similar runtime performance as that shown in Figure~\ref{fig:rcaddres}. 

	\begin{figure}[!b]
		\centering
		\includegraphics[scale=.75]{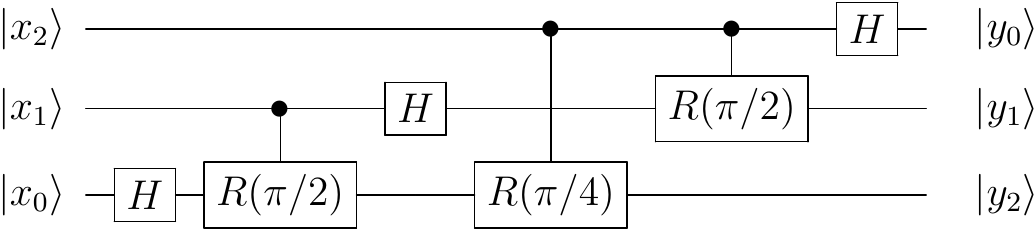}
		\vspace{5pt}
		\caption{\label{fig:qft} The three-qubit QFT circuit.}
	\end{figure}

\ \\\noindent
{\bf Quantum Fourier transform (QFT) circuits}. Our third benchmark set consists 
of circuits that implement the $n$-qubit QFT, which computes the discrete 
Fourier transform of the amplitudes in the input quantum state. 
Let $\ket{x_1x_2\ldots x_n}$, $x_i\in\{0,1\}$ be a computational-basis state
and $\mathbf{x}_{1,2,\ldots,m}=\sum_{k=1}^mx_k 2^{-k}$.
The action of the QFT on this input state can be expressed as:
	\begin{small}
	\begin{align}
		\ket{x_1\ldots x_n} = \frac{1}{\sqrt{2^n}} &
		\left(\ket{0}+e^{2i\pi\cdot\mathbf{x}_{n}}\ket{1}\right)\otimes
		\left(\ket{0}+e^{2i\pi\cdot\mathbf{x}_{n-1,n}}\ket{1}\right)\otimes \notag \\
		&\cdots\otimes\left(\ket{0}+e^{2i\pi\cdot\mathbf{x}_{1,2,\dots, n}}\ket{1}\right)
	\end{align}
	\end{small}
The QFT is used in many quantum algorithms, notably Shor's 
factoring and discrete logarithm algorithms. Such circuits are composed of a network
of Hadamard and controlled-$R(\alpha)$ gates, where $\alpha=\pi/2^k$ and $k$ is the
distance over which the gate acts. The three-qubit QFT circuit is shown in Figure~\ref{fig:qft}.
In general, the first qubit requires one Hadamard gate, the next qubit requires a Hadamard and a controlled-$R(\alpha)$ gate, and each following qubit requires an additional 
controlled-$R(\alpha)$ gate. Summing up the number of gates gives $O(n^2)$ for
an $n$-qubit QFT circuit. Figure~\ref{fig:qftres} shows average runtime and memory usage for both 
\quipu and \qplt on QFT instances for $n=\{10,12,\ldots,20\}$. \quipu runs approximately 
$10\times$ faster than \qplt on average and consumes about $96\%$ less memory.
For these benchmarks, we observed that the number of states in our multiframe
data structure was $2^{n-1}$. This is because controlled-$R(\alpha)$ gates produce biased superpositions (Section~\ref{sec:sframes}) that cannot be effectively compressed using our 
coalescing procedure. Therefore, as Figure~\ref{fig:qftres} shows, the runtime 
and memory requirements of both \quipu and \qplt grow
exponentially in $n$ for QFT instances. However, \quipu scales to 
$24$-qubit instances while \qplt scales to only $18$ qubits. The 
multi-threaded version of \quipu exhibited roughly a $2\times$ speedup
and used a comparable amount of memory on a four-core Xeon server.  

We compared \quipu to a straightforward 
implementation of the state-vector model using an array of complex amplitudes. 
Such non-compact data structures can be streamlined to simulate most quantum 
gates (e. g., Hadamard, controlled-$R(\alpha)$) with limited 
runtime overhead, but scale to only around $30$ qubits due to poor memory scaling. 
Our results showed that \quipu was approximately $3\times$ slower than an array-based
implementation when simulating QFT instances. However, such implementations
cannot take advantage of circuit structure and, unlike \quipu and \qplt, 
do not scale to instances of stabilizer and arithmetic circuits 
with $>30$ qubits (Figures~\ref{fig:stabcomp}~and~\ref{fig:rcaddres}).

	\begin{figure}[!t]
	\begin{tabular}{lclc}
		\hspace{-.4cm}
		\rotatebox{90}{\hspace{6mm}\small Runtime (secs)} 
		& \hspace{-.8cm}
		\includegraphics[scale=.36]{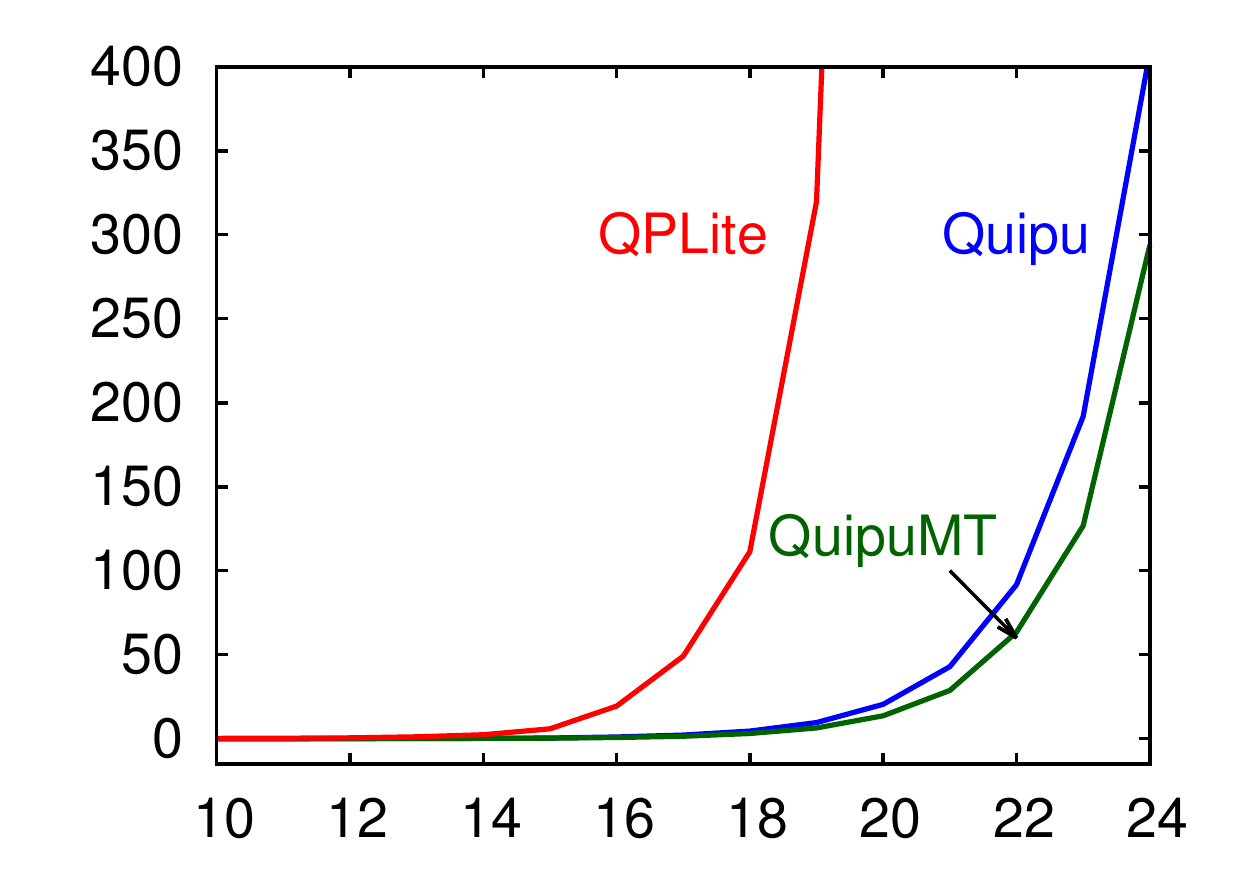} 
		&
		\hspace{-.6cm}
		\rotatebox{90}{\hspace{3mm}\small Peak memory (MB)}
		& \hspace{-.85cm}
		\includegraphics[scale=.36]{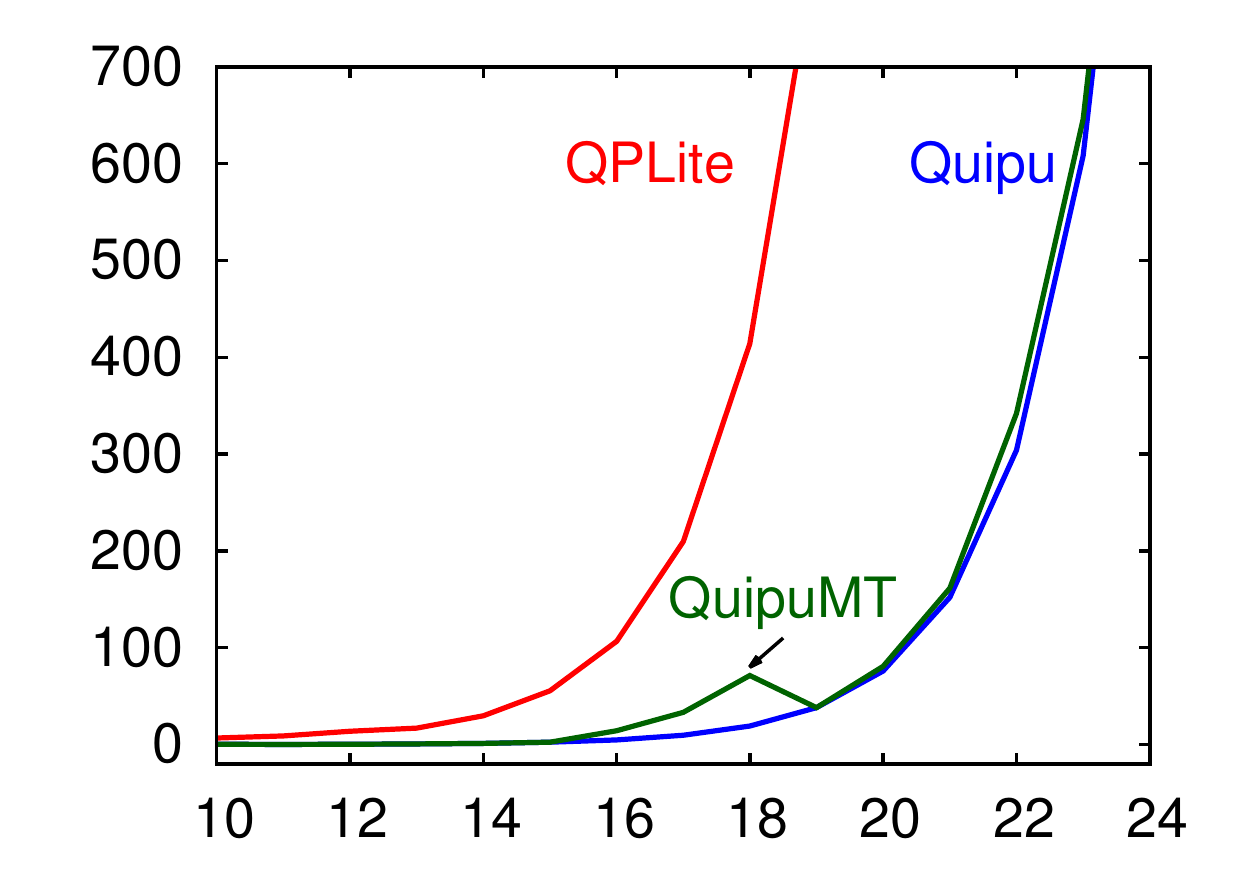}  \\ 
		\multicolumn{4}{c}{\small $n$-qubit QFT circuit}
	\end{tabular}
	\parbox{\linewidth}{
	\caption{\label{fig:qftres} Average runtime and memory needed by \quipu 
	(single and multi-threaded) and \qpro to simulate $n$-qubit 
	QFT circuits, which contain $n(n+1)/2$ gates. We used the $\ket{11\ldots 1}$ 
	input state for all benchmarks.
        }}
        \vspace{-8pt}
	\end{figure}
	
	\begin{table*}[!t]
		\centering\footnotesize
		\parbox[b]{\linewidth}{
		\caption{\label{tab:ftbench} Average time and memory needed by \quipu 
		and \qplt to simulate several quantum FT circuits. The second column shows the QECC used to encode 
		$k$ logical qubits into $n$ physical qubits. Benchmarks with ($^*$) use 
		the $5$-qubit DiVincenzo/Shor code~\cite{DiVincen} instead of the $3$-qubit bit-flip code. 
		We used the $\ket{00\ldots 0}$ input state for all benchmarks. The top numbers from each
		row correspond to direct simulation of Toffoli gates (Section~\ref{sec:sframes}) and the 
		bottom numbers correspond to simulation via the decomposition from Figure~\ref{fig:toffdecomp}.
		Shaded rows are Clifford circuits for mod-exp. 
		}
		\vspace{-5pt}
		}
		\begin{tabular}{|c|c|c|c|c||c|c|c|c||c|c|c|} \hline
			\sc fault-tolerant & \sc qecc & \sc num.  
			& \multicolumn{2}{c||}{\sc num. gates}  
			& \multicolumn{2}{c|}{\sc runtime (secs)} 
			& \multicolumn{2}{c||}{\sc memory (MB)} & \multicolumn{2}{c|}{\sc max size($\Psi$)}
			\\ 
			\sc circuit & $[n, k]$ & \sc qubits & \sc cliff \& meas. & \sc non-cliff. & \qplt 
			& \quipu & \qplt & \quipu & \sc single $\cF$ & \sc multi $\cF$ \\ 
			\hline \hline
			\multirow{2}{10mm}{toffoli$^*$} & \multirow{2}{7mm}{$[15,3]$} & \multirow{2}{3.5mm}{$45$} & $155$  & $15$  & $43.68$ & $\bf 0.20$  & $98.45$ & $\bf 12.76$ & $2816$ & $32$ \\
			 & & & $305$  & $90$  & $72.45$ & $\bf 0.83$  & $137.02$ & $\bf 12.78$ & $8192$ & $32$ \\ 
			\hline
			\multirow{2}{10mm}{halfadd$^*$} & \multirow{2}{7mm}{$[15,3]$} & \multirow{2}{3.5mm}{$45$} & $160$  & $15$  & $43.80$ & $\bf 0.20$  & $94.82$ & $\bf 12.76$ & $2816$ & $32$ \\
			& & & $310$  & $90$  & $75.05$ & $\bf 0.84$  & $137.03$ & $\bf 12.78$ & $8192$ & $32$ \\ 
			\hline
			\multirow{2}{10mm}{fulladd$^*$} & \multirow{2}{7mm}{$[20,4]$} & \multirow{2}{3.5mm}{$80$} & $320$  & $30$  & $84.96$ & $\bf 0.88$  & $91.86$ & $\bf 12.94$ & $2816$ & $32$ \\ 
			& & & $620$  & $180$  & $1173.48$ & $\bf 1.61$  & $139.44$ & $\bf 13.52$ & $16384$ & $32$ \\
			\hline
			\multirow{2}{10mm}{$2^x$mod$15$} & \multirow{2}{7mm}{$[18,6]$} & \multirow{2}{3.5mm}{$81$} & $396$  & $36$  & $4.81$hrs & $\bf 1.48$ & $11.85$ & $\bf 12.96$ & $22528$ & $64$ \\ 
			& & & $756$  & $216$  & $>24$hrs & $\bf 6.08$ & $118.17$ & $\bf 14.23$ & $2^{22}$ & $64$ \\
			\hline
			\rowcolor[gray]{0.85} $4^x$mod$15^*$ & $[30,6]$ & $30$ & $30$ & $0$ & $0.01$ & $\bf< 0.01$ & $6.14$ & $\bf 12.01$ & $1$ & $1$ \\
			\hline
			\multirow{2}{10mm}{$7^x$mod$15$}  & \multirow{2}{7mm}{$[18,6]$} & \multirow{2}{3.5mm}{$81$} & $402$  & $36$  & $11.25$hrs & $\bf 1.52$ & $12.41$ & $\bf 13.29$ & $22528$ & $64$ \\
			& & & $762$  & $216$  & $>24$hrs & $\bf 4.98$ & $134.05$ & $\bf 14.77$ & $2^{22}$ & $64$ \\
			\hline
			\multirow{2}{10mm}{$8^x$mod$15$}  & \multirow{2}{7mm}{$[18,6]$} & \multirow{2}{3.5mm}{$81$} & $399$  & $36$  & $11.37$hrs & $\bf 1.52$ & $12.48$ & $\bf 13.29$ &  $22528$ & $64$ \\
			& & & $759$  & $216$  & $>24$hrs & $\bf 6.08$ & $135.18$ & $\bf 14.77$ &  $2^{22}$ & $64$ \\
			\hline
			\rowcolor[gray]{0.85} $11^x$mod$15^*$  & $[30,6]$ & $30$ & $25$   & $0$   & $0.02$ & $\bf< 0.01$ & $6.14$ & $\bf 12.01$ & $1$ & $1$ \\ \hline
			\multirow{2}{10mm}{$13^x$mod$15$} & \multirow{2}{7mm}{$[18,6]$} & \multirow{2}{3.5mm}{$81$} & $399$  & $36$  & $11.28$hrs & $\bf 1.56$ & $11.85$ & $\bf 12.25$ & $22528$ & $64$  \\
			& & & $759$  & $216$  & $>24$hrs & $\bf 4.64$ & $135.23$ & $\bf 14.62$ & $2^{22}$ & $64$ \\
			\hline
			\rowcolor[gray]{0.85} $14^x$mod$15^*$  & $[30,6]$ & $30$  & $40$   & $0$   & $0.02$ & $\bf< 0.01$ & $6.14$ & $\bf12.01$ & $1$ & $1$ \\
			\hline
		\end{tabular}
	\end{table*}
	
\ \\\noindent
{\bf Fault-tolerant (FT) circuits}. Our next benchmark set consists of 
circuits that, in addition to preparing encoded quantum states, implement
procedures for performing FT quantum operations \cite{DiVincen, NielChu, Preskill}.
FT operations limit the propagation errors from one qubit in a
QECC-register (the block of qubits that encodes a logical qubit) to 
another qubit in the same register, and a single faulty 
gate damages at most one qubit in each register.
One constructs FT stabilizer circuits by executing each {\color{blue} Clifford} gate
transversally\footnote{In a transversal operation, 
the $i^{th}$ qubit in each QECC-register interacts only with the $i^{th}$ 
qubit of other QECC-registers~\cite{Gottes98, NielChu, Preskill}.} 
across QECC-registers
as shown in Figure~\ref{fig:transcnot}. 
{\color{blue} Non-Clifford gates need to be implemented 
using a FT architecture that often requires ancilla qubits, measurements 
and correction procedures conditioned on measurement outcomes. 
Figure~\ref{fig:fttoff} shows a circuit that implements a 
FT-Toffoli operation~\cite{Preskill}. Each line represents 
a $5$-qubit register based on the DiVincenzo/Shor\footnote{The DiVincenzo/Shor $5$-qubit code 
functions successfully in the presence of both bit-flip and phase-flip 
errors even if they occur during correction procedures~\cite{DiVincen}.} 
code, and each gate is applied transversally. The state 
$\ket{cat} = (\ket{0^{\otimes 5}}+\ket{1^{\otimes 5}})/\sqrt{2}$
is obtained using a stabilizer subcircuit (not shown).
The arrows point to the set of gates that is applied 
if the measurement outcome is $1$; no action is taken otherwise.
Controlled-{$Z$} gates are implemented as 
$H_jCNOT_{i,j}H_j$ with control $i$ and target $j$. 
$Z$ gates are implemented as $P^2$. 
}
	\begin{figure}[!t]
		\centering
		\includegraphics[scale=.65]{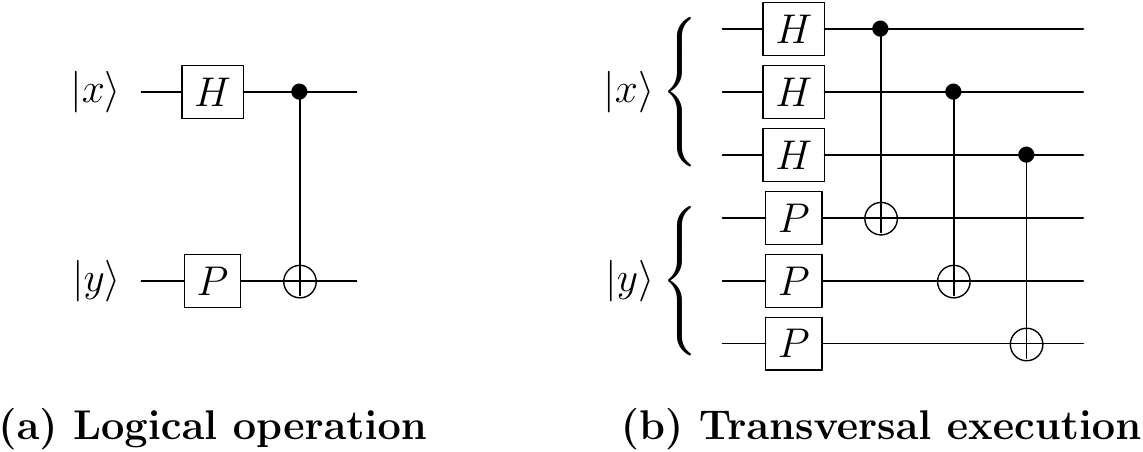}
		\caption{\label{fig:transcnot} Transversal implementation 
		of a stabilizer circuit acting on three-qubit QECC registers.}
	\end{figure}
	
	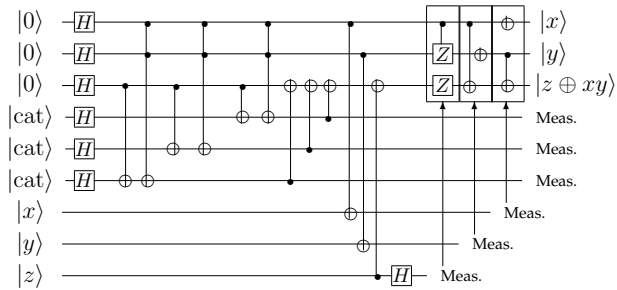
\begin{figure}[!t]
		\centering\hspace{-14mm}
		\scalebox{.75}[.75]{\input{fttoffdm}}
		\vspace{-10pt}
		\caption{\label{fig:fttoff} Fault-tolerant implementation of a Toffoli gate.
		}
	\end{figure}

We implemented FT benchmarks for the half-adder 
and full-adder circuits (Figure~\ref{fig:hfadders}) as well as for computing $f(x) = b^x$mod $15$.
Each circuit from Figure \ref{fig:modexp} implements $f(x)$ with a particular coprime base 
value $b$ as a $(2, 4)$ look-up table (LUT).\footnote{A $(k, m)$-LUT 
takes $k$ read-only input bits and $m > \log_2 k$ ancilla bits. For 
each $2^k$ input combination, an LUT produces a pre-determined $m$-bit value, e.g., 
a $(2,4)$-LUT is defined by values $(1,2,4,8)$ or $(1,4,1,4)$.}
The Toffoli gates in all our FT benchmarks are implemented
using the architecture from Figure~\ref{fig:fttoff}.  
Since FT-Toffoli operations require $6$ ancilla registers,
a circuit that implements $t$ FT-Toffolis using a 
$k$-qubit QECC, requires $6tk$ ancilla qubits. Therefore,
to compare with \qplt, we used the $3$-qubit bit-flip code~\cite[Ch. 10]{NielChu} 
instead of the more robust $5$-qubit code in our 
larger benchmarks. 
Our results in Table~\ref{tab:ftbench} show
that \quipu is typically faster than \qplt by several
orders of magnitude and consumes 
$8\times$ less memory for the {\em toffoli},
{\em half-adder} and {\em full-adder} benchmarks.
{\color{blue}
For FT benchmarks that consist of 
stabilizer circuits (shaded rows),
the QuIDD representation remains compact and 
utilizes half as much memory as our frame representation.
}

	\begin{figure}[!b]
		\centering\Large
		\includegraphics[scale=.75]{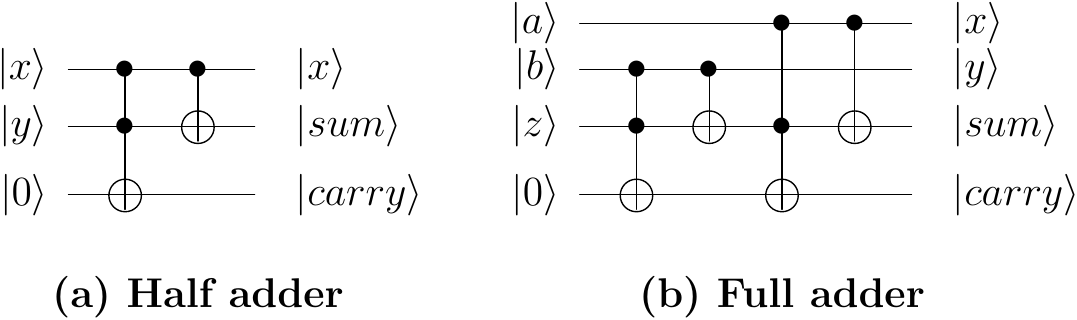}
		\caption{\label{fig:hfadders} Adder circuits implemented in our benchmarks. 
		}
	\end{figure}
	
Table~\ref{tab:ftbench} also shows that our coalescing 
technique is very effective as the maximum size of the
stabilizer-state superposition is much smaller
when multiple frames are used. 
{\color{blue} Since the total number of states 
observed is relatively small, the multithreaded version of 
\quipu exhibited similar	runtime and memory requirements 
for these benchmarks.

	\begin{figure}[!t]	
		\vspace{2mm}
		\hspace{-2mm}
		\centering
		\includegraphics[scale=.55]{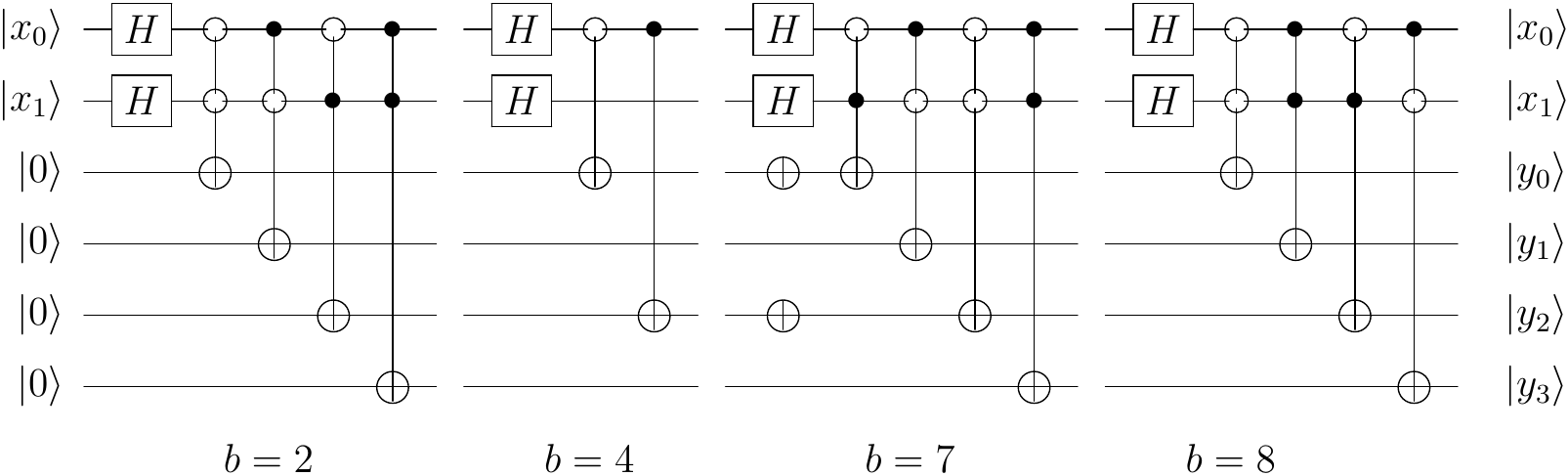}
		\vspace{-4mm}
		\caption{\label{fig:modexp} Mod-exp with $M=15$ implemented 
		as $(2, 4)$-LUTs \cite{MarkovSaeedi} for several coprime base values.
		Negative controls are shown with hollow circles. 
		}
	\end{figure}

\ \\ 
\noindent
{\bf Technology-dependent circuits}. Section~\ref{sec:sframes} outlines how
\quipu supports primitive gate libraries, especially for 
quantum-optical systems where Clifford gates are considered primitive~\cite{Lin}.
Therefore, to simulate FT circuits for
photonic systems, it suffices to decompose $TOFF$ gates into sequences of 
Hadamard, CNOT and $T$ gates as shown in Figure~\ref{fig:toffdecomp}.
Table~\ref{tab:ftbench} reports simulations of our FT
benchmarks using such decompositions.
Since the total number of gates is larger, \quipu is roughly 
$4\times$ slower as compared to direct simulation of $TOFF$ gates.
\qplt takes $>24$ hours to simulate several of these benchmarks
while \quipu takes only several seconds since the majority of the gates 
introduced by the decomposition from Figure~\ref{fig:toffdecomp} 
are Clifford gates.

\begin{figure}[!b]
	\centering\Large
	\includegraphics[scale=.7]{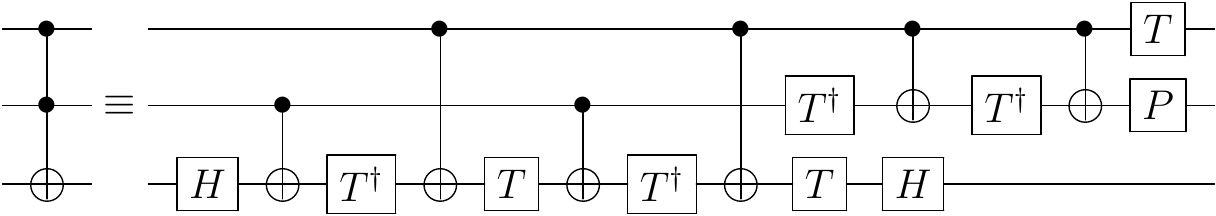}
	\caption{\label{fig:toffdecomp} The Toffoli gate and its decomposition into one-qubit 
	and CNOT gates\cite[Figure 4.9]{NielChu}.}
\end{figure}
}

\section{Conclusions} \label{sec:conclude}


In this work, we developed new techniques for quantum-circuit
simulation based on superpositions of stabilizer states,
avoiding shortcomings in prior work~\cite{AaronGottes}.
To represent such superpositions compactly, we designed a 
new data structure called a {\em stabilizer frame}.
We implemented stabilizer frames and relevant algorithms 
in our software package \quipu. 
Current simulators based on the stabilizer formalism, such as \chp, are 
limited to simulation of stabilizer circuits. Our results show that \quipu performs
asymptotically as fast as \chp on stabilizer circuits with
a linear number of measurement gates, but
simulates certain quantum arithmetic circuits in polynomial time and space 
for input states consisting of equal superpositions of 
computational-basis states. In contrast, \qpro takes exponential time
on such instances. We simulated quantum 
Fourier transform and quantum fault-tolerant circuits with \quipu, and 
the results demonstrate that our stabilizer-based technique 
leads to orders-of-magnitude improvement in runtime and memory 
as compared to \qpro. 
While our technique uses more sophisticated mathematics and quantum-state 
modeling, it is significantly easier to implement and optimize. In particular,
our multithreaded implementation of \quipu exhibited a $2\times$ speed up on 
a four-core server. 

%
{\color{blue}
\noindent
{\bf Future Directions}. The work in~\cite{Yamashita} describes 
an equivalence-checking method for quantum circuits based on the 
notion of a {\em reversible miter}~--~a counterpart of miter circuits
used in equivalence-checking of digital circuits.
An attractive direction for future work is deriving new {\em Clifford 
miters}~--~linear combinations of Clifford operators that 
represent a specific quantum circuit. Clifford miters can speed up 
formal verification by exploiting similarities in circuits and 
the fast equivalence-checking algorithms from~\cite{AaronGottes, Garcia}.

Section~\ref{sec:sframes} described how \quipu incorporates
quantum machine descriptions in the form of primitive gate libraries 
to simulate technology-dependent circuits.
This method can benefit from new decompositions for 
library gates into linear combinations of Pauli or Clifford operators.
Such decompositions can be obtained {\em on the fly} when 
simulating original gates one-by-one in sequence. They can also be 
precomputed and used for a {\em compiled} version of the original circuit, 
where scheduling can be optimized for parallelism and architecture
constraints. 
}

\vspace{2mm}

\noindent
{\bf Acknowledgements}. This work was sponsored in part by the Air 
Force Research Laboratory under agreement FA8750-11-2-0043.

\end{document}

%% file: fttoffdm.tex
\setlength{\unitlength}{.8pt}
\begin{picture}(335,200)

\put(10,174){\makebox(20,12){\large$\ket{0}$}}
\put(10,154){\makebox(20,12){\large$\ket{0}$}}
\put(10,134){\makebox(20,12){\large$\ket{0}$}}
\put(10,114){\makebox(20,12){\large$\ket{\rm cat}$}}
\put(10,94){\makebox(20,12){\large$\ket{\rm cat}$}}
\put(10,74){\makebox(20,12){\large$\ket{\rm cat}$}}
\put(10,54){\makebox(20,12){\large$\ket{x}$}}
\put(10,34){\makebox(20,12){\large$\ket{y}$}}
\put(10,14){\makebox(20,12){\large$\ket{z}$}}

\put(40,180){\line(1,0){8}}
\put(48,174){\framebox(12,12){$H$}}
\put(40,160){\line(1,0){8}}
\put(48,154){\framebox(12,12){$H$}}
\put(40,140){\line(1,0){8}}
\put(48,134){\framebox(12,12){$H$}}
\put(40,140){\line(1,0){8}}
\put(48,74){\framebox(12,12){$H$}}
\put(40,140){\line(1,0){8}}
\put(48,94){\framebox(12,12){$H$}}
\put(40,140){\line(1,0){8}}
\put(48,114){\framebox(12,12){$H$}}

\put(60,180){\line(1,0){275}}
\put(60,160){\line(1,0){214}}
\put(286,160){\line(1,0){49}}
\put(60,140){\line(1,0){214}}
\put(286,140){\line(1,0){49}}

\put(42,80){\line(1,0){6}}
\put(60,80){\line(1,0){270}}

\put(80,80){\circle{8}}
\put(80,140){\circle*{4}}
\put(80,140){\line(0,-1){64}}

\put(94,179){\circle*{4}}
\put(93,180){\line(0,-1){104}}
\put(94,159){\circle*{4}}
\put(94,80){\circle{8}}


\put(42,100){\line(1,0){6}}
\put(60,100){\line(1,0){270}}
\put(110,100){\circle{8}}

\put(112,139){\circle*{4}}
\put(111,140){\line(0,-1){44}}

\put(130,179){\circle*{4}}
\put(129,180){\line(0,-1){84}}
\put(130,159){\circle*{4}}
\put(130,100){\circle{8}}


\put(42,120){\line(1,0){6}}
\put(60,120){\line(1,0){270}}
\put(154,120){\circle{8}}

\put(154,139){\circle*{4}}
\put(153,140){\line(0,-1){24}}

\put(170,179){\circle*{4}}
\put(170,180){\line(0,-1){64}}
\put(170,159){\circle*{4}}
\put(170,120){\circle{8}}

\put(184,140){\circle{8}}
\put(184,80){\line(0,1){64}}
\put(184,79){\circle*{4}}
\put(197,140){\circle{8}}
\put(196,100){\line(0,1){44}}
\put(196,99){\circle*{4}}
\put(209,140){\circle{8}}
\put(208,120){\line(0,1){24}}
\put(208,119){\circle*{4}}


\put(40,60){\line(1,0){270}}
\put(40,40){\line(1,0){250}}
\put(40,20){\line(1,0){208}}

\put(222,179){\circle*{4}}
\put(221,180){\line(0,-1){124}}
\put(222,59){\circle{8}}

\put(230,159){\circle*{4}}
\put(230,160){\line(0,-1){124}}
\put(230,39){\circle{8}}

\put(239,19){\circle*{4}}
\put(238,20){\line(0,1){124}}
\put(239,140){\circle{8}}

\put(248,14){\framebox(12,12){$H$}}
\put(260,20){\line(1,0){10}}

\put(335,74){\makebox(30,12){\footnotesize  ~Meas.}}
\put(335,94){\makebox(30,12){\footnotesize  ~Meas.}}
\put(335,114){\makebox(30,12){\footnotesize  ~Meas.}}
\put(275,14){\makebox(30,12){\footnotesize  ~Meas.}}
\put(295,34){\makebox(30,12){\footnotesize  ~Meas.}}
\put(315,54){\makebox(30,12){\footnotesize  ~Meas.}}

\put(280,179){\circle*{4}}
\put(279,180){\line(0,-1){14}}
\put(274,154){\framebox(12,12){$Z$}}
\put(274,134){\framebox(12,12){$Z$}}
\put(270,130){\framebox(20,60){}}
\put(280,26){\vector(0,1){104}}

\put(297,179){\circle*{4}}
\put(297,180){\line(0,-1){44}}
\put(297,139){\circle{8}}
\put(304,159){\circle{8}}
\put(303.5,156){\line(0,1){8}}
\put(291,130){\framebox(20,60){}}
\put(300,46){\vector(0,1){84}}

\put(321,179){\circle{8}}
\put(320,176){\line(0,1){8}}
\put(321,159){\circle*{4}}
\put(320,160){\line(0,-1){24}}
\put(321,140){\circle{8}}
\put(311,130){\framebox(20,60){}}
\put(320,66){\vector(0,1){64}}

\put(340,174){\makebox(20,12){\large$\ket{x}$}}
\put(340,154){\makebox(20,12){\large$\ket{y}$}}
\put(340,134){\makebox(40,12){\large$~\ket{z \oplus xy}$}}

\end{picture}